\documentclass{amsart}

\usepackage{amssymb,amsmath,amsthm, color}
\usepackage{inputenc}
\usepackage{url}
\usepackage{amsmath}
\usepackage{float}
\usepackage{pst-node}
\usepackage{tikz}

\usepackage{graphicx}
\usepackage{tikz-cd} 

\newtheorem{theorem}{Theorem}[section]
\newtheorem{lemma}[theorem]{Lemma}

\newtheorem{corollary}[theorem]{Corollary}
\newtheorem{proposition}[theorem]{Proposition}
\theoremstyle{definition}
\newtheorem{definition}[theorem]{Definition}
\newtheorem{example}[theorem]{Example}

\theoremstyle{remark}
\usepackage{amssymb}
\newtheorem{remark}[theorem]{Remark}
\usepackage{enumerate}

\newcommand{\G}{\mathcal G}

\newcommand{\F}{\mathbb F}

\newcommand{\cyc}{\operatorname{Cyc}}

\newcommand{\Z}{\mathbb Z}
\newcommand{\C}{\mathcal C}

\newcommand{\GL}{\mathrm{GL}}
\newcommand{\PGL}{\mathrm{PGL}}

\newcommand{\seq}[1]{\textrm{{\boldmath $#1$}}}

\DeclareMathOperator{\ord}{o}
\DeclareMathOperator{\lcm}{lcm}

\newcommand{\mqureshi}[1]{{#1}}

\newcommand{\doublespace}

\author[C. Qureshi]{Claudio Qureshi}
\address{Instituto de Matem\'atica y Estad{\'\i}stica Rafael Laguardia, Facultad de Ingenier{\'\i}a, Universidad de la Rep\'ublica, Montevideo 11300, Uruguay.}
\email{cqureshi@gmail.com}

\author[L. Reis]{Lucas Reis}
\address{Departamento de Matem\'{a}tica, Universidade Federal de Minas Gerais, Belo Horizonte, MG, 30270-901, Brazil.}
\email{lucasreismat@gmail.com}


\begin{document}

\title{On the functional graph of the power map over finite groups}

\begin{abstract}
In this paper we study the description of the functional graphs associated with the power maps over finite groups. \mqureshi{We present a structural result which describes the isomorphism class of these graphs for abelian groups and also for flower groups, which is a special class of non abelian groups introduced in this paper. Unlike the abelian case where all the trees associated with periodic points are isomorphic, in the case of flower groups we prove that several different classes of trees can occur. The class of central trees (i.e. associated with periodic points that are in the center of the group) are in general non-elementary and a recursive description is given in this work. Flower groups include many non abelian groups such as dihedral and generalized quaternion groups, and the projective general linear group of order two over a finite field. In particular, we provide improvements on past works regarding the description of the dynamics of the power map over these groups.}
\end{abstract}

\maketitle

\section{Introduction}
Given a pair $(f, S)$ of a finite set $S$ and a map $f:S\to S$, we can associate with it a graph $\mathcal G(f/S)$, called the functional graph of $f$ over $S$. This is the directed graph with vertex set $V=\{s\,|\, s\in S\}$ and directed edges $\{s\to f(s)\,|\, s\in S\}$. The functional graph encodes the dynamics of $f$ over $S$. For instance, the $f$-orbit $\{s, f(s), f^{(2)}(s), \ldots\}$ of an element $s\in S$  is described by a path in $\mathcal G(f/S)$. Moreover, an element $s\in S$ is $f$-periodic if and only if it belongs to a cycle of $\mathcal G(f/S)$. One of the motivation of studying finite dynamical systems is due to their applications such as integer factorization methods in cryptography \cite{P75,WZ98} and pseudo-random number generators \cite{BBS}. The description of $\mathcal G(f/S)$ has been considered for many algebraic structures $S$ and well behaved maps $f$. In many cases, the functional graphs turns out to have many remarkable properties such as regularity on the indegrees and symmetries on its cycles that allow us to obtain a partial or complete description of their structure. See~\cite{Gassert14, MV88,PR18, PMMY01, QP15, QP18, QPM17, Rogers96, T05, 
VS04} for a rich source of results regarding these issues and \cite{MPQ19} for a survey including some applications of dynamical systems over finite fields.

When $S=G$ is a finite group, it is natural to consider the power map $\varphi_t:G\to G$ with $\varphi_t(g)=g^t$ (or $tg$ if $G$ is written additively), where $t$ is an integer. When $G$ is cyclic, the graph $\mathcal G(\varphi_t/G)$ is completely described in~\cite{QP15} \mqureshi{and the trees attached to the periodic points are described by a non-increasing sequence of positive integers. Trees constructed in this way are called elementary and they turn out to be useful to describe the non-periodic part of the dynamic of many interesting class of maps such as Chebyshev polynomials over finite fields, Redei functions and some maps related to certain endomorphism of elliptic curves over finite fields \cite{QR19}. In \cite{deKlerketal} the authors describe the graph $\mathcal G(\varphi_t/G)$ for special cases of abelian groups. In section \ref{subsection:powermap-abeliangroup} we obtain a complete description of this graph for general abelian groups, mainly using the explicit description given in~\cite{QP15} for the cyclic case and using some results of~\cite{QR19}. The trees attached to the periodic points in this case (which are all isomorphic) are expressed as a tensor product of elementary trees. When the group is non abelian the power map $\varphi_t$ is no longer an homomorphism and the trees attached to the periodic points are no longer isomorphic and in general it is a difficult problem to describe them. The trees attached to periodic points in the center of the group are isomorphic as we show in Section \ref{subsection:powermap-finitegroup}, these trees are called {\it central trees} and are, in general, non-elementary. Regarding the problem of describing the graph $G(\varphi_t/G)$} for non abelian groups $G$, only results for special families are known. More specifically, in \cite{ahmad}, \cite{a2} and \cite{deng} the cases where $G$ is a dihedral, a generalized quaternion and some semidirect product of cyclic groups are explored, respectively. However the digraphs are not explicitly determined there. Instead, only partial descriptions of such graph are given like the distribution of indegrees, cycle lengths and number of cycles. Further results of this kind are given in~\cite{larsen} for special classes of finite groups.

In this paper we introduce the class of {\it flower groups}, which contains the dihedral and generalized quaternion groups, and also some semidirect product of cyclic groups. Our main result is the description of the functional graph structure of the power map over flower groups. Such description is given in two parts: in Theorem \ref{thm:main} we provide a complete description of the cyclic structure and also of the non central trees, and in Theorem \ref{thm:tree} we provide some properties of central trees which allow us to obtain a explicit description of such trees for very special cases of interest.

Our paper is organized as follows. In Section \ref{Section:Preliminaries} we introduce the basic notation and operations on graphs and review the structure theorem for the functional graph of the power map on cyclic groups. Then, using results from~\cite{QR19}, we show how to extend this structure theorem for finite abelian groups. In the last part of this section we provide some general results on the functional graph of the power map on arbitrary finite groups that is further used to prove our main results. In Section \ref{Section:main} we introduce the notion of flower groups, obtain some properties about these groups and also prove our main results. In Section \ref{Section:applications} we apply our main results to the explicit description of functional graphs associated to some non abelian groups.

\section{Preliminaries} \label{Section:Preliminaries}
This section provides some background material and minor results. We start by fixing some basic notation and reviewing the structural description of the functional graph of the power map on cyclic groups. Then we show how to extend this structural description for abelian groups. At the end of this section we provide some general results on the functional graph of the power map over arbitrary finite groups.

Along the paper we use the letters $G, H$ to denote finite groups and $C$ to denote cyclic groups. For a positive integer $m$, $\mathcal C_m$ denotes the cyclic group of order $m$. We use multiplicative notation for the operation of these groups unless otherwise specified. The ring of integers modulo $d$ is denoted by $\Z_d$ and their multiplicative group of units is denoted by $\Z_d^*$. The Euler totient function is denoted by $\varphi(d)$ and the (multiplicative) order of $t \in \Z_d^*$ is denoted by $\ord_d(t)$. For a group $G$, the centralizer of $S\subseteq G$ is denoted by $C_G(S)$ and the center of $G$ is denoted by $Z(G)=C_G(G)$.

\subsection{Basics on functional graphs}
We use the same terminology as in \cite{QP15} and \cite{QR19}. It is a well known fact that the connected components of functional graphs are composed of a cycle and each vertex of this cycle is the root of a tree (the direction of the edges is from leaves to the root). Several interesting classes of functions over finite fields or finite structures studied in the literature present certain regularity that is reflected in some symmetry conditions on their functional graphs. We denote by $\cyc(m,T)$ a directed graph composed by an $m$-cycle (i.e. a cycle of length $m$) where every node in this cycle is the root of a tree isomorphic to $T$. When every connected component of a functional graph $\G(f/S)$ is of this form we say that this functional graph is {\it regular}, in this case there are integers $m_1,\ldots,m_k$ and rooted trees $T_1,\ldots, T_k$ such that $\G(f/S)=\bigoplus_{i=1}^{k}\cyc(m_i,T_i)$, where the circled plus symbol denotes a disjoint union of graphs. When the rooted tree $T$ has a unique vertex we write $T=\bullet$. The $m$-cycle $\cyc(m,\bullet)$ is denoted by $\cyc(m)$ and the rooted tree $T$ with a loop in the root $\cyc(1,T)$ is denoted by $\{T\}$. The notation $G=k\times H$ means that $G=\bigoplus_{i=1}^{k}H_i$ with each $H_i$ isomorphic to the graph $H$. A forest is a disjoint union of rooted trees. Given a forest $G=\bigoplus_{i=1}^{k}T_i$, we denote by $\langle G \rangle$ the rooted tree with $k$ children where these children are roots of trees isomorphic to $T_1,\ldots, T_k$. We consider the empty graph $\emptyset$ with the property that $\langle \emptyset \rangle = \bullet$. Given rooted trees $T_1=\langle G_1 \rangle, \cdots, T_k=\langle G_k \rangle$ where each $G_i$ is a forest  we define $\sum_{i=1}^{k}T_i := \langle \bigoplus_{i=1}^{k}G_i \rangle$ (i.e. the sum of rooted trees is a new rooted tree which is obtained by identifying all the roots). For a tree $T=\langle G \rangle$ and $k\in \mathbb{Z}^{+}$ we denote $k\cdot T = \langle k \times G \rangle$ (note that $k\cdot T \neq k \times T$ if $k>1$). \\

The vertices belonging to the cycles of the functional graph $\G(f/S)$ correspond to the periodic points of $f$. Given a point $x_0 \in S$, the least natural number $\delta=\delta(x_0)$ such that $f^{(\delta)}(x_0)$ is a periodic point is called the {\it preperiod} of $x_0$ (the periodic points corresponds to points $x_0$ with $\delta(x_0)=0$). Regarding the functional graph $\G(f/S)$, the number $\delta(x_0)$ equals the depth of $x_0$ in the rooted tree which it belongs to (i.e. the distance of $x_0$ to the root).\\

There is a simple way to associate with each non-increasing finite sequence of positive integers $\seq{v}=(\nu_1,\nu_2,\ldots, \nu_d)$ (i.e. $\nu_1\geq \nu_2 \geq \cdots, \nu_d \geq 1$) a rooted tree $T_{\seq{v}}$ defined recursively as follows:

\begin{equation}\label{TreeAssociatedEq}
\left\{ \begin{array}{l}
           T_{\seq{v}}^{0}= \bullet,   \\
          G_{\seq{v}}^{k}= \nu_k \times T_{\seq{v}}^{k-1} \oplus
                 \bigoplus_{i=1}^{k-1}(\nu_{i}-\nu_{i+1})\times T_{\seq{v}}^{i-1}
                  \textrm{ and } T_\seq{v}^{k}=\langle G_{\seq{v}}^{k} \rangle \textrm { for } 1\leq k <d, \\
    G_{\seq{v}} =  (\nu_d-1) \times T_{\seq{v}}^{d-1} \oplus
                 \bigoplus_{i=1}^{d-1}(\nu_{i}-\nu_{i+1})\times T_{\seq{v}}^{i-1}
                  \textrm{ and } T_\seq{v}= \langle G_{\seq{v}} \rangle.
         \end{array}
\right.\end{equation}

Trees associated with non-increasing sequences as above are called {\it elementary trees}, see \cite{QR19} for more details on elementary trees. Figure \ref{inductivo} shows the inductive process to contruct $T_\seq{v}$ for a $4$-term sequence.


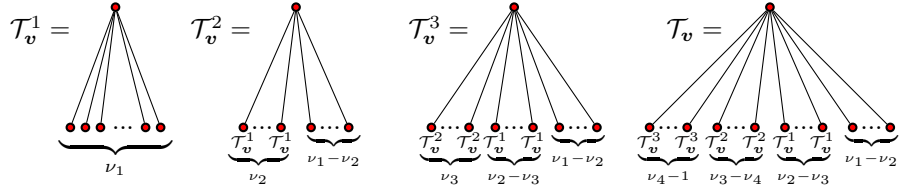
\begin{figure}[H]
\begin{center}
\begin{tikzpicture}
  [scale=1, place/.style={circle,draw=black,thick,fill=red,
                 inner sep=0pt,minimum size=1mm}]
  \node (T1) at (-1,-0.3) {$\mathcal{T}_{\seq{v}}^{1}=$};
  \node (R) at ( 0,0) [place] {};
  \node (1) at ( -0.6,-1.6) [place] {};
  \node (2) at ( -0.4,-1.6) [place] {};
  \node (3) at ( -0.2,-1.6) [place] {};
  \node (ptos) at ( 0.1, -1.6) {...};
  \node (4) at ( 0.4,-1.6) [place] {};
  \node (5) at ( 0.6,-1.6) [place] {};
  \node (underbrace) at (0, -2) {$\underbrace{\hspace{1.5cm}}_{\nu_1}$};
\foreach \x/\y in {R/1, R/2, R/3, R/4, R/5}
	\draw (\x) to (\y);
	
\begin{scope}[xshift=2.4cm, yshift=0cm]
  \node (T2) at (-1,-0.3) {$\mathcal{T}_{\seq{v}}^{2}=$};
  \node (R) at ( 0,0) [place] {};
  \node (1) at ( -0.7,-1.6) [place] {};
  \node (1T) at (-0.7, -1.8) {\tiny $\mathcal{T}_{\seq{v}}^{1}$};
  \node (ptos1) at ( -0.45,-1.6) {...};
  \node (2) at ( -0.2,-1.6) [place] {};
  \node (2T) at (-0.2,-1.8) {\tiny $\mathcal{T}_{\seq{v}}^{1}$}; 
  \node (underbrace1) at (-0.5, -2.15) {\scriptsize $\underbrace{\hspace{0.8cm}}_{\nu_2}$};
  \node (3) at ( 0.2, -1.6) [place] {};
  \node (ptos2) at (0.45,-1.6) {...};
  \node (4) at ( 0.7,-1.6) [place] {};
  \node (underbrace2) at (0.5, -1.9) {\scriptsize $\underbrace{\hspace{0.7cm}}_{\nu_1-\nu_2}$};
\foreach \x/\y in {R/1, R/2, R/3, R/4}
	\draw (\x) to (\y);
\end{scope}

\begin{scope}[xshift=5.3cm, yshift=0cm]
  \node (T3) at (-1,-0.3) {$\mathcal{T}_{\seq{v}}^{3}=$};
  \node (R) at (0,0) [place] {};
  \node (1) at (-1.1,-1.6) [place] {};
  \node (1T) at (-1.1, -1.8) {\tiny $\mathcal{T}_{\seq{v}}^{2}$};
  \node (ptos1) at (-0.85,-1.6) {...};
  \node (2) at ( -0.6,-1.6) [place] {};
  \node (2T) at (-0.6,-1.8) {\tiny $\mathcal{T}_{\seq{v}}^{2}$}; 
  \node (underbrace1) at (-0.85, -2.15) {\scriptsize $\underbrace{\hspace{0.8cm}}_{\nu_3}$};
  \node (3) at ( -0.25, -1.6) [place] {};
  \node (3T) at (-0.25, -1.8) {\tiny $\mathcal{T}_{\seq{v}}^{1}$}; 
  \node (ptos2) at (0,-1.6) {...};
  \node (4) at ( 0.25,-1.6) [place] {};
  \node (4T) at (0.25,-1.8) {\tiny $\mathcal{T}_{\seq{v}}^{1}$};
  \node (underbrace2) at (0, -2.15) {\scriptsize $\underbrace{\hspace{0.7cm}}_{\nu_2-\nu_3}$};
  \node (5) at ( 0.6, -1.6) [place] {};
  \node (ptos3) at (0.85,-1.6) {...};
  \node (6) at ( 1.1,-1.6) [place] {};
  \node (underbrace3) at (0.85, -1.9) {\scriptsize $\underbrace{\hspace{0.7cm}}_{\nu_1-\nu_2}$};
\foreach \x/\y in {R/1, R/2, R/3, R/4, R/5, R/6}
	\draw (\x) to (\y);
\end{scope}

\begin{scope}[xshift=8.7cm, yshift=0cm]
  \node (T) at (-1,-0.3) {$\mathcal{T}_{\seq{v}}=$};
  \node (R) at (0,0) [place] {};
  \node (1) at (-1.6,-1.6) [place] {};
  \node (1T) at (-1.6,-1.8) {\tiny $\mathcal{T}_{\seq{v}}^{3}$};
  \node (ptos1) at (-1.35,-1.6) {...};
  \node (2) at (-1.1,-1.6) [place] {};
  \node (2T) at (-1.1,-1.8) {\tiny $\mathcal{T}_{\seq{v}}^{3}$}; 
  \node (underbrace1) at (-1.35, -2.15) {\scriptsize $\underbrace{\hspace{0.8cm}}_{\nu_4-1}$};
  \node (3) at ( -0.7, -1.6) [place] {};
  \node (3T) at (-0.7, -1.8) {\tiny $\mathcal{T}_{\seq{v}}^{2}$}; 
  \node (ptos2) at (-0.45,-1.6) {...};
  \node (4) at (-0.2,-1.6) [place] {};
  \node (4T) at (-0.2,-1.8) {\tiny $\mathcal{T}_{\seq{v}}^{2}$};
  \node (underbrace2) at (-0.45, -2.15) {\scriptsize $\underbrace{\hspace{0.7cm}}_{\nu_3-\nu_4}$};  
  \node (5) at (0.2, -1.6) [place] {};
  \node (5T) at (0.2, -1.8) {\tiny $\mathcal{T}_{\seq{v}}^{1}$}; 
  \node (ptos3) at (0.45,-1.6) {...};
  \node (6) at ( 0.7,-1.6) [place] {};
  \node (7T) at (0.7,-1.8) {\tiny $\mathcal{T}_{\seq{v}}^{1}$};
  \node (underbrace2) at (0.45, -2.15) {\scriptsize $\underbrace{\hspace{0.7cm}}_{\nu_2-\nu_3}$};
  \node (7) at (1.1, -1.6) [place] {};
  \node (ptos3) at (1.35,-1.6) {...};
  \node (8) at (1.6,-1.6) [place] {};
  \node (underbrace3) at (1.35, -1.9) {\scriptsize $\underbrace{\hspace{0.7cm}}_{\nu_1-\nu_2}$};
\foreach \x/\y in {R/1, R/2, R/3, R/4, R/5, R/6, R/7, R/8}
	\draw (\x) to (\y);
\end{scope}

\end{tikzpicture}
\end{center}
\caption{This figure (taken from \cite{QP15}) illustrates the inductive 
definition of $T_\seq{v}$ for a $4$-term sequence $\seq{v}=(\nu_1,\nu_2,\nu_3,\nu_4)$. A node labelled by a rooted tree $T$ indicates that it is the root of a tree isomorphic to $T$.}
\label{inductivo}
\end{figure}


Elementary trees appear in a wide range of context when we describe the dynamics of different classes of maps, see for example \cite{Gassert14, PR18, QP15,QP18, QPM17,QR19,Ugolini18}. They will play a key role also in our description of the functional graph associated with the power map over flower groups. An important special case of elementary trees $T_\seq{v}$ is when the sequence $\seq{v}$ is a multiplicative chain (i.e. when each term of the sequence divides the previous term), these are the trees that appear in the functional graph of power maps on abelian groups.


A rooted tree $T$ is {\it homogeneous} if it can be written as $T=\langle \bigoplus_{i=0}^{d-1}  n_i \times T_i \rangle$ where $T_i$ is a rooted tree with depth $i$ for $0\leq i <d$ (i.e. all its subtrees with the same depth are isomorphic). In particular elementary trees are homogeneous since the subtrees $T_{\seq{v}}^i$ of $T_{\seq{v}}$ associated with the $d$-term sequence $\seq{v}$ has depth $i$, for $0\leq i < d $. The following operation will be usefull to describe trees in the functional graph associated with the power map on some finite non abelian groups.

\begin{definition}
Let $T =\langle \bigoplus_{i=0}^{d-1}  n_i \times T_i \rangle$ be an homogeneous rooted tree where each $T_i$ has depth $i$ for $0\leq i <d$ and $S$ be a rooted tree. For $0\leq j<d $, the $j$-sum of $T$ and $S$ is given by $T +_{j} S = \langle  \bigoplus_{i\neq j}  n_i \times T_i \oplus (n_j-1)\times T_j \oplus (T_j +S)  \rangle$.
\end{definition}

We note that $T +_{j} S$ is obtaining by replacing one of the depth-$j$-subtree $T_j$ of $T$ with $T_j+S$. The resulting rooted tree will be not an homogeneous tree in general.

\subsection{The power map on finite cyclic groups}\label{subsection:powermap-finitegroup}

Let $\mathcal C_n$ be the (multiplicative) cyclic group of order $n$ and $\varphi_t:\mathcal C_n \to \mathcal C_n$ be the power map $x\to x^t$. The trees attached to the cyclic points in $\G(\varphi_t/\mathcal C_n)$ are described by the iterated gcd of $n$ relative to $t$ (also called $\nu$-series in \cite{QP15}), denoted by $\gcd_t(n)$, and defined as follows\footnote{In \cite{QP15} the iterated gcd of $n$ relative to $t$ is called $\nu$-series generated by $n$ and $t$, denoted by $n(t)$, and defined only in the case when each prime divisor of $n$ divides $t$. In order to avoid possible confusion we decide to use the alternative notation $\gcd_t(n)$.}: $\gcd_t(n)=(1)$ if $\gcd(t,n)=1$, otherwise $\gcd_t(n)=(\nu_1,\ldots,\nu_D)$ where 
$$\nu_1=\gcd(t,n), \quad \nu_{i+1}=
  \gcd\left(t,\frac{n}{\nu_1\cdots \nu_i}\right), \mbox{ for } i \geq 1,$$ 
and $D$ is the least positive integer such that $\nu_{D+1}=1$. It is easy to see that if $\nu:=\nu_1\cdots \nu_D$ then $\nu \mid n$, $\gcd_t(n)=\gcd_t(\nu)$ and $\omega:=n/\nu$ is the greatest divisor of $n$ that is relatively prime with $t$ (see \cite{QP15} for more details). The following proposition gives an explicit description for the functional graph $\G(\varphi_t/\mathcal C_n)$.


\begin{proposition}[\cite{QPM17}, Proposition 2.1]\label{prop:cyc}
Let $n=\nu\omega$, where $\omega$ is the greatest divisor of $r$ that is relatively prime with $t$. Let $\gcd_{t}(\nu)=(\nu_1,\nu_2,\ldots,\nu_{D})$ be the iterated gcd of $\nu$ relative to $t$ and $T_{\gcd_{t}(\nu)}$ be the elementary tree associated with this sequence. Then $\mathcal C_n$ has exactly $\omega$ elements that are $\varphi_t$-periodic and the following isomorphism formula holds: 
\begin{equation}\label{IsomorphismFormulaEq}
\G(\varphi_t/\mathcal C_n)= \bigoplus_{d\mid \omega}
\left(\frac{\varphi(d)}{\ord_d(t)} \times \cyc\left(\ord_d(t),T_{\gcd_{t}(\nu)}\right) \right).
\end{equation}
Moreover, the tree $T_{\gcd_{t}(\nu)}$ has $\nu$ vertices and depth $D$.
\end{proposition}

\subsection{The power map on finite abelian groups}\label{subsection:powermap-abeliangroup}
Next we show how formula (\ref{IsomorphismFormulaEq}) can be extended to general abelian groups. Note that if \seq{u} and \seq{v} are non-increasing finite sequences which differ only possibly in the last terms that are all equal to $1$, then the corresponding elementary trees are equal (i.e. $T_{\seq{u}}=T_{\seq{v}}$). In these case we say that these sequences are {\it equivalent}. The product  $\seq{uv}$ of two non-increasing sequences is the coordinatewise product (substituting the smallest sequence by other equivalent in such a wat that both sequences have the same length).
For what follows, $\otimes$ denotes the {\em tensor product} of digraphs.

\begin{lemma}[\cite{QR19}, Lemma 3.4]\label{lem:FDS}
For any maps of finite sets $f:X\to X$ and $g:Y\to Y$ we have the following graph isomorphism:
$$\G(f\times g/ X\times Y) \cong \G(f/X)\otimes \G(g/Y),$$
where $f\times g:X\times Y\to X\times Y$ is the map $f\times g\ (x, y)=(f(x), g(y))$. 
\end{lemma}

\begin{lemma}[\cite{QR19}, Lemma 3.5]\label{lem:FDS-2}
Let $r\in \mathbb{Z}^{+}$ and $T$ be a rooted tree. The following isomorphism holds: $$\mathrm{Cyc}(r,T) \cong \mathrm{Cyc}(r)\otimes \{T\}. $$
\end{lemma}

\begin{lemma}[\cite{QR19}, Prop. 2.10]\label{lem:tree-1}
For any non-decreasing sequences \seq{u}, \seq{v} we have $\{T_\seq{u}\}\otimes \{T_\seq{v}\}=\{T_{\seq{uv}}\}$.
\end{lemma}

\begin{lemma}\label{lem:CyclesTensor}
Let $r_1,\ldots, r_k$ be positive integers. We have $$\bigotimes_{i=1}^{k}\mathrm{Cyc}(r_i) = \frac{r_1r_2\cdots r_k}{\lcm(r_1,r_2,\ldots,r_k)} \times \mathrm{Cyc}\left(\lcm(r_1,r_2,\ldots,r_k)\right). $$
\end{lemma}

\begin{proof}
Consider the maps $s_i:\Z_{r_i}\to \Z_{r_i}$, $s_i(x)=x+1$. Then $\mathrm{Cyc}(r_i)=\G(s_i / \Z_{r_i})$ and $\bigotimes_{i=1}^{k}\mathrm{Cyc}(r_i)=\G(s_1\times \cdots \times s_k, \Z_{r_1}\times \cdots \times \Z_{r_k})$ which is a union of disjoint cycles. Each one of this cycles is in correspondence with the cosets of $H:=\langle (1,1,\ldots,1) \rangle$ in $\Z_{r_1}\times \cdots \times \Z_{r_k}$ and each one of these cosets has $|H|=\lcm(r_1,r_2,\ldots,r_k)$ elements. 
\end{proof}

Now consider any (multiplicative) abelian group $G$ and the map $\varphi_t: G \to G$, $\varphi_t(g):=g^t$.  By the fundamental theorem of finite abelian groups, there is an isomorphism $\eta:G\to \mathcal C_{r_1}\times \cdots \times \mathcal C_{r_k}$. Since $\eta\circ \varphi_t = \varphi_t\circ \eta$ we have that $\eta$ induces an isomorphism between the functional graph of $\varphi_t$ over $G$ and over the direct product of cyclic groups, so we can assume $G=\mathcal C_{r_1}\times \cdots \times \mathcal C_{r_k}$. If $\seq{d}$ and $\seq{\omega}$ are $k$-term sequences, $\seq{d}\mid \seq{\omega}$ means $d_i\mid \omega_i$ for every $1\leq i \leq k$.

\begin{proposition}
Let $G$ be an abelian group and write $G=\mathcal C_{r_1}\times \cdots \times \mathcal C_{r_k}$, where $\mathcal C_r$ denotes a cyclic group of order $r$. Let $r_i=\nu_i \omega_i$ where $\omega_i$ is the greatest divisor of $r_i$ that is relatively prime with $t$, $\seq{\nu}:=(\nu_1,\ldots,\nu_k)$, $\seq{\omega}:=(\omega_1,\ldots, \omega_k)$ and $\gcd_t(\nu):=\prod_{i=1}^{k}\gcd_t(\nu_i)$. For $\seq{d}=(d_1,\ldots, d_k)$ define $\varphi(\seq{d}):=\prod_{i=1}^{k}\varphi(d_i)$ and $\ord_{\seq{d}}(t)=\lcm\{\ord_{d_i}(t): 1\leq i \leq k \}$. Then G has exactly $\prod_{i=1}^{k}\omega_i$ elements that are $\varphi_t$-periodic and the following isomorphism formula holds:
\begin{equation}\label{IsomorphismFormulaEq2}
\G(\varphi_t/G)= \bigoplus_{\seq{d}\mid\seq{\omega}} \frac{\varphi(\seq{d})}{\ord_{\seq{d}}(t)} \times \cyc\left( \ord_{\seq{d}}(t), T_{\gcd_t(\seq{\nu})}\right) 
\end{equation}
\end{proposition}

\begin{proof}
If formula (\ref{IsomorphismFormulaEq2}) holds, the number of $\varphi_t$-periodic points in $G$ is $\sum_{\seq{d}\mid \seq{\omega}}\varphi(\seq{d})= \prod_{i=1}^{k}\left( \sum_{d_i\mid \omega_i}\varphi(d_i)\right)$ $= \prod_{i=1}^{k}\omega_i$. To prove the formula we use Proposition \ref{prop:cyc} and Lemma \ref{lem:FDS} to obtain: $$\G(\varphi_t/G)= \bigotimes_{i=1}^{k} \G(\varphi_t/\mathcal C_{r_i}) = \bigotimes_{i=1}^{k}\bigoplus_{d_i\mid \omega_i}
\left( \frac{\varphi(d_i)}{\ord_{d_i}(t)} \times \cyc\left(\ord_{d_i}(t),T_{\gcd_{t}(\nu_i)}\right) \right).$$
By the commutativity of the tensor product and Lemma \ref{lem:FDS-2} we have:
\begin{eqnarray*}
\G(\varphi_t/G) &=& \bigoplus_{\seq{d}\mid\seq{\omega}}\bigotimes_{i=1}^{k}\left( \frac{\varphi(d_i)}{\ord_{d_i}(t)}\times \cyc\left( \ord_{d_i}(t),T_{\gcd_{t}(\nu_i)}  \right)\right)  \\ 
&=& \bigoplus_{\seq{d}\mid\seq{\omega}}   \left(\prod_{i=1}^{k} \frac{\varphi(d_i)}{\ord_{d_i}(t)} \right)\times  \left(\bigotimes_{i=1}^{k}\cyc\left( \ord_{d_i}(t)\right) \right)  \otimes \left( \bigotimes_{i=1}^{k}   T_{\gcd_{t}(\nu_i)}  \right)
\end{eqnarray*}
Finally, using Lemmas \ref{lem:tree-1}, \ref{lem:CyclesTensor} and again Lemma \ref{lem:FDS-2} we obtain the desired formula.

\end{proof}

\begin{example}    
Consider the power map $\varphi_{14}$ over the group $\Z_{91}^{*}\cong \C_{6}\times \C_{12}$ of invertible elements modulo $91$. Figure~\ref{fig:abelian} shows the functional graph of this map. In this case $k=2$, $r_1=6$, $r_2=12$, $\seq{\nu}=(2,4)$, $\seq{\omega}=(3,3)$ and $\gcd_{14}(\seq{\nu})= \gcd_{14}(2)\cdot \gcd_{14}(4)= (2)\cdot (2,2) = (4,2)$. Then,
\begin{eqnarray*}
\mathcal{G}(\varphi_{14}/G) &=& \bigoplus_{d_1\mid 3 \atop d_2\mid 3} \frac{\varphi(d_1)\varphi(d_2)}{\lcm\{o_{d_1}(14),o_{d_2}(14)\}}\times \cyc\left(\lcm\{o_{d_1}(14),o_{d_2}(14)\},\mathcal{T}_{(4,2)}\right)\\
&=& \{\mathcal{T}_{(4,2)}\} \oplus 4\times\cyc\left(2, \mathcal{T}_{(4,2)} \right).
\end{eqnarray*}
\begin{figure}[H]
\begin{center}
\begin{tikzpicture}
  [scale=1, place/.style={circle,draw=black,thick,fill=red,
                 inner sep=0pt,minimum size=1mm}]
\draw (0,0.25) circle (2.5mm);
  \node (1) at ( 0,0) [place] {};
  \node (2) at ( 0.5,-0.5) [place] {};
  \node (3) at ( 0, -0.5) [place] {};
  \node (4) at ( -0.5,-0.5) [place] {};
  \node (5) at ( 0.5,-1) [place] {};
  \node (6) at ( 0.25,-1) [place] {};
  \node (7) at ( -0.25,-1) [place] {};
  \node (8) at ( -0.5,-1) [place] {};
\foreach \x/\y in {1/2, 1/3, 1/4, 3/5, 3/6, 3/7, 3/8}
	\draw (\x) to (\y);

\begin{scope}[xshift=2cm, yshift=0cm]
  \node (1) at ( 0,0.5) [place] {};
  \node (2) at ( 0.5,1) [place] {};
  \node (3) at ( 0, 1) [place] {};
  \node (4) at ( -0.5,1) [place] {};
  \node (5) at ( 0.5,1.5) [place] {};
  \node (6) at ( 0.25,1.5) [place] {};
  \node (7) at ( -0.25,1.5) [place] {};
  \node (8) at ( -0.5,1.5) [place] {};
  \draw (0,0.25) circle (2.5mm);
  \node (11) at ( 0,0) [place] {};
  \node (12) at ( 0.5,-0.5) [place] {};
  \node (13) at ( 0, -0.5) [place] {};
  \node (14) at ( -0.5,-0.5) [place] {};
  \node (15) at ( 0.5,-1) [place] {};
  \node (16) at ( 0.25,-1) [place] {};
  \node (17) at ( -0.25,-1) [place] {};
  \node (18) at ( -0.5,-1) [place] {};
\foreach \x/\y in {11/12, 11/13, 11/14, 13/15, 13/16, 13/17, 13/18}
	\draw (\x) to (\y);
\foreach \x/\y in {1/2, 1/3, 1/4, 3/5, 3/6, 3/7, 3/8}
	\draw (\x) to (\y);
\end{scope}

\begin{scope}[xshift=4cm, yshift=0cm]
  \node (1) at ( 0,0.5) [place] {};
  \node (2) at ( 0.5,1) [place] {};
  \node (3) at ( 0, 1) [place] {};
  \node (4) at ( -0.5,1) [place] {};
  \node (5) at ( 0.5,1.5) [place] {};
  \node (6) at ( 0.25,1.5) [place] {};
  \node (7) at ( -0.25,1.5) [place] {};
  \node (8) at ( -0.5,1.5) [place] {};
  \draw (0,0.25) circle (2.5mm);
  \node (11) at ( 0,0) [place] {};
  \node (12) at ( 0.5,-0.5) [place] {};
  \node (13) at ( 0, -0.5) [place] {};
  \node (14) at ( -0.5,-0.5) [place] {};
  \node (15) at ( 0.5,-1) [place] {};
  \node (16) at ( 0.25,-1) [place] {};
  \node (17) at ( -0.25,-1) [place] {};
  \node (18) at ( -0.5,-1) [place] {};
\foreach \x/\y in {11/12, 11/13, 11/14, 13/15, 13/16, 13/17, 13/18}
	\draw (\x) to (\y);
\foreach \x/\y in {1/2, 1/3, 1/4, 3/5, 3/6, 3/7, 3/8}
	\draw (\x) to (\y);
\end{scope}

\begin{scope}[xshift=6cm, yshift=0cm]
  \node (1) at ( 0,0.5) [place] {};
  \node (2) at ( 0.5,1) [place] {};
  \node (3) at ( 0, 1) [place] {};
  \node (4) at ( -0.5,1) [place] {};
  \node (5) at ( 0.5,1.5) [place] {};
  \node (6) at ( 0.25,1.5) [place] {};
  \node (7) at ( -0.25,1.5) [place] {};
  \node (8) at ( -0.5,1.5) [place] {};
  \draw (0,0.25) circle (2.5mm);
  \node (11) at ( 0,0) [place] {};
  \node (12) at ( 0.5,-0.5) [place] {};
  \node (13) at ( 0, -0.5) [place] {};
  \node (14) at ( -0.5,-0.5) [place] {};
  \node (15) at ( 0.5,-1) [place] {};
  \node (16) at ( 0.25,-1) [place] {};
  \node (17) at ( -0.25,-1) [place] {};
  \node (18) at ( -0.5,-1) [place] {};
\foreach \x/\y in {11/12, 11/13, 11/14, 13/15, 13/16, 13/17, 13/18}
	\draw (\x) to (\y);
\foreach \x/\y in {1/2, 1/3, 1/4, 3/5, 3/6, 3/7, 3/8}
	\draw (\x) to (\y);
\end{scope}

\begin{scope}[xshift=8cm, yshift=0cm]
  \node (1) at ( 0,0.5) [place] {};
  \node (2) at ( 0.5,1) [place] {};
  \node (3) at ( 0, 1) [place] {};
  \node (4) at ( -0.5,1) [place] {};
  \node (5) at ( 0.5,1.5) [place] {};
  \node (6) at ( 0.25,1.5) [place] {};
  \node (7) at ( -0.25,1.5) [place] {};
  \node (8) at ( -0.5,1.5) [place] {};
  \draw (0,0.25) circle (2.5mm);
  \node (11) at ( 0,0) [place] {};
  \node (12) at ( 0.5,-0.5) [place] {};
  \node (13) at ( 0, -0.5) [place] {};
  \node (14) at ( -0.5,-0.5) [place] {};
  \node (15) at ( 0.5,-1) [place] {};
  \node (16) at ( 0.25,-1) [place] {};
  \node (17) at ( -0.25,-1) [place] {};
  \node (18) at ( -0.5,-1) [place] {};
\foreach \x/\y in {11/12, 11/13, 11/14, 13/15, 13/16, 13/17, 13/18}
	\draw (\x) to (\y);
\foreach \x/\y in {1/2, 1/3, 1/4, 3/5, 3/6, 3/7, 3/8}
	\draw (\x) to (\y);
\end{scope}

\end{tikzpicture}
\end{center}

\caption{The graph $\{\mathcal{T}_{(4,2)}\} \oplus 4 \times \operatorname{Cyc}\left( 2,\mathcal{T}_{(4,2)}  \right)$.}
\label{fig:abelian}
\end{figure}
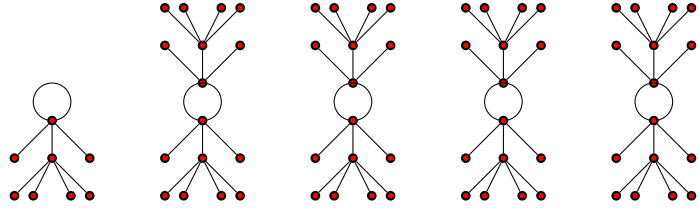


\end{example}

\subsection{Some results on the power map over finite groups}

Let $G$ be a finite group and let $d$ be a divisor of $|G|$. We denote by $G[d]$ the group generated by the elements $g\in G$ such that $g^d=1$. In the next lemma we consider the factorization $|G|=\nu \omega$ where $\omega$ is the greatest divisor of $|G|$ that is relatively prime with $t$. We note that if $g\in G$ verifies $g^{t^n}=1$ then the order of $g$ is a divisor of $\nu$ since $\gcd(t^n,|G|)$ divides $\nu$. In particular, the tree attached to $1$ in $\mathcal G(\varphi_t/G)$, is contained in $G[\nu]$.

\begin{definition}
An element $g\in G$ is $\varphi_t$-periodic if there exists a positive integer $n$ such that $g^{t^n}=\varphi_t^{(n)}(g)=g$. Moreover, the {\em preperiod} of $g\in G$ under $\varphi_t$ is the least integer $i\ge 0$ such that  $\varphi_t^{(i)}(g)$ is $\varphi_t$-periodic.
\end{definition}

\mqureshi{
\begin{definition}
The {\em central tree} of $\mathcal{G}(\varphi_t/G)$ is the rooted tree attached to neutral element $1\in G$ and it is denoted by $\mathcal{T}_{t}(G)$.
\end{definition}}

\begin{proposition}\label{lem:centre}
Let $G$ be a finite group with identity $1$ and consider the power map $\varphi_t:G\to G$ with $g\mapsto g^t$. Let $|G|=\nu \omega$ where $\omega$ is the greatest divisor of $|G|$ that is relatively prime with $t$. If $h\in G$ is $\varphi_t$-periodic and $h$ is in the centralizer of $G[\nu]$, then the rooted tree attached to $h$ in $\mathcal G(\varphi_t/G)$ is isomorphic to \mqureshi{the central tree $\mathcal{T}_{t}(G)$}.
\end{proposition}

\begin{proof}
Suppose that $s$ is the least positive integer such that $\varphi_t^{(s)}(h)=h$, i.e., $h^{t^s}=h$. Let $\mathcal T_h$ be the rooted tree attached to $h$ in $\mathcal G(\varphi_t/G)$. Set $\tau: \mqureshi{\mathcal{T}_{t}(G)}\to \mathcal T_h$ with $\tau(g)=gh^{t^{\delta(g)(s-1)}}$, where $\delta(g)$ is the preperiod of $g$ under $\varphi_t$. As previously remarked, each element of $\mqureshi{\mathcal{T}_{t}(G)}$ is in $G[\nu]$.

{\bf Claim.} {\em The map $\tau$ is well defined and preserves the preperiod, i.e, $\tau(g)\in \mathcal T_h$ and $\delta(\tau(g))=\delta(g)$ for every $g\in \mqureshi{\mathcal{T}_{t}(G)}$. }

We proceed by induction on $n=\delta(g)$. If $n=0$, then $g=1$ and $\tau(1)=h\in \mathcal T_h$ with $\delta(h)=0$. If $n=1$, then $g\ne 1$, $g^t=1$ and $\tau(g)=gh^{t^{s-1}}$. In particular, since $h$ is in the centralizer of $G[\nu]$ and $g\in \mqureshi{\mathcal{T}_{t}(G)}$, we have that $\varphi_t(\tau(g))=g^th^{t^{s}}=h$. Since $g\ne 1$, it follows that $\tau(g)=gh^{t^{s-1}}\ne h^{t^{s-1}}$, the unique $\varphi_t$-periodic element in the set $\varphi_t^{-1}(h)$. The latter implies that $\tau(g)\in \mathcal T_h$ and $\delta(\tau(g))=1$.
Suppose that $\tau(g)\in\mathcal T_h$ and $\delta(\tau(g))=\delta(g)$ whenever $g\in \mqureshi{\mathcal{T}_{t}(G)}$ and $\delta(g)= n$ for some $n\ge 1$ and let $g\in \mqureshi{\mathcal{T}_{t}(G)}$ with $\delta(g)=n+1$. Since $n\ge 1$, we necessarily have that $\varphi_t(g)\in \mqureshi{\mathcal{T}_{t}(G)}$ and $\delta(\varphi_t(g))=n$. By induction hypothesis, $\tau(\varphi_t(g))\in \mathcal T_h$ and
$\delta(\tau(\varphi_t(g)))=n$. Since $\mqureshi{\mathcal{T}_{t}(G)}\subseteq G[\nu]$, $h$ is in the centralizer of $G[\nu]$ and $n(s-1)\equiv (n+1)(s-1)+1\pmod{s}$ we have that 
\begin{equation}\label{EqTauPreserveAdj}
\tau(\varphi_t(g)) =g^th^{t^{n(s-1)}}=\left(gh^{t^{(n+1)(s-1)}}\right)^{t}=\varphi_t(\tau(g)).
\end{equation}
In particular, $\tau(g)$ is in the set $\varphi_t^{-1}(f)$ for some element $f\in \mathcal T_h$ of preperiod $n$. Since $n\ge 1$, the latter implies that $\tau(g)$ is in $\mathcal T_h$ and has preperiod $n+1$. The proof of the claim is complete.

From the claim, $\tau$ is well defined and preserves the preperiod under $\varphi_t$. By the same reasoning, the map $\tau^*:\mathcal T_h\to \mqureshi{\mathcal{T}_{t}(G)}$ that sends the element $g\in \mathcal T_h$ to the element $gh^{-t^{\delta(g)(s-1)}}$ is well defined and preserves the preperiod under $\varphi_t$. It is direct to verify that $\tau$ and $\tau^*$ are the compositional inverses of each other, hence $\tau$ is a bijection. By Equation (\ref{EqTauPreserveAdj}), we have that $\tau\circ \varphi_t=\varphi_t\circ \tau$ and then $\tau$ preserves adjacency.

\end{proof}

\begin{corollary} Let $h\in G$ be a $\varphi_t$-periodic element. If $h \in Z(G)$ then the rooted tree attached to $h$ in $\mathcal G(\varphi_t/G)$ is isomorphic to the \mqureshi{central tree $\mathcal{T}_t(G)$}.
\end{corollary}

\section{On flower groups}\label{Section:main}
Given a group $G$, a cyclic subgroup $H$ is a $\mu$-subgroup of $G$ if $H$ is not contained in any cyclic subgroup of $G$ other than $H$ itself. We have the following definition.

\begin{definition}
Let $G$ be a finite noncyclic group and let $S=\{C_1,C_2,\ldots,C_k\}$ be the collection of its $\mu$-subgroups. The group $G$ is a flower group if there exists a subgroup $C_0$ of $G$ such that $C_i\cap C_j=C_0$ for $1\leq i<j\leq k$. The subgroup $C_0$ is the {\em pistil} of $G$ and the elements of $S$ are the {\em petals} of $G$. We define the type of $G$ as $(c_0;c_1,\ldots,c_k)$ where $c_i:=|C_i|$ for $0\leq i \leq k$.
\end{definition}

In the following proposition we provide equivalent conditions for a group $G$ to be a flower group.

\begin{proposition}\label{prop:count}
For a finite noncyclic group $G$ and a cyclic subgroup $C_0$ of $G$, the following are equivalent:

\begin{enumerate}[(i)]
\item $G$ is a flower group with pistil $C_0$;
\item for any $g\in G\setminus C_0$, there exists a unique $\mu$-subgroup $C$ such that $g\in C$ and $C_0\subseteq C$;
\item there exist $k\ge 1$ and distinct $\mu$-subgroups $C_1, \ldots, C_k$ of $G$ such that $C_i\cap C_j=C_0$ for any $1\le i<j\le k$, verifying
$$\sum_{i=1}^k|C_i|-(k-1)|C_0|=|G|.$$
\end{enumerate}
\end{proposition}

\begin{proof}
Let $S$ be the collection of the $\mu$-subgroups of $G$. It is direct to verify that $S$ covers $G$, i.e., $G=\bigcup_{C\in S}C$. For the implication (i)$\rightarrow$ (ii), we observe that if $g\in G$ belongs to two distinct $\mu$-subgroups $C_1, C_2$ of $G$, then $g\in C_1\cap C_2=C_0$. The implication (ii)$\rightarrow$(iii) follows by a simple counting argument. It remains to prove that (iii)$\rightarrow$ (i). If (iii) holds, it follows that $G$ is the disjoint union of the sets $C_i\setminus C_0$ and $C_0$. In particular, since each $C_i$ is a $\mu$-subgroup of $G$ and $S$ covers $G$, we have that $S=\{C_1, \ldots, C_k\}$. Therefore, $G$ is a flower group with pistil $C_0$.
\end{proof}

In the following proposition we provide some basic properties of flower groups.

\begin{proposition}\label{prop:properties}
Let $G$ be a flower group with pistil $C_0$ and center $Z(G)$. Then the following hold:

\begin{enumerate}[(i)]
\item If $C=\langle g \rangle$ is a $\mu$-subgroup of $G$ then $C\subseteq C_G(g)$. In particular, $C_0 \subseteq Z(G)$.
\item If $H$ is any noncyclic subgroup of $G$, $H$ is a flower group with pistil $H\cap C_0$;
\end{enumerate}
\end{proposition}

\begin{proof}
\begin{enumerate}[(i)]
\item Since $g\in C_G(g)$ we have $C=\langle g \rangle\subseteq C_G(g)$. Observe that every element of $G$ is contained in at least one $\mu$-subgroup of $G$. Therefore, the intersection of all the $\mu$-subgroups is contained in the intersection of all the centralizers, i.e., $C_0\subseteq Z(G)$. 
\item It suffices to prove that the $\mu$-subgroups of $H$ are of the form $H \cap C$ with $C$ a $\mu$-subgroup of $C$. Let $C'$ be a $\mu$-subgroup of $H$. Since $C'$ is a cyclic subgroup of $G$, it is contained in a $\mu$-subgroup $C$ of $G$. We have that  $C'\subseteq H\cap C$ and $H\cap C$ is a cyclic subgroup of $H$, hence $C'=H\cap C$.
\end{enumerate} 
\end{proof}

%
%
%


The following result provides a special class of flower groups $G$ whose pistil $C_0$ coincides with $Z(G)$.




\begin{proposition}\label{prop:com}
Let $G$ be a finite non abelian group such that for every element $g\in G\setminus Z(G)$, there exists a unique $\mu$-subgroup $C_g^*$ of the centralizer $C_G(g)$ containing $g$ and $Z(G)$. Then $G$ is a flower group with pistil $Z(G)$ and the set of petals of $G$ equals $\{C_g^*\,|\, g\in G\setminus Z(G)\}$.
\end{proposition}

\begin{proof}
We observe that, for every $g\in G\setminus Z(G)$,  the inclusion $Z(G)\subseteq C_g^*$ holds and, in particular, $Z(G)$ is cyclic. From Proposition~\ref{prop:count}, it suffices to prove that  for every $g\in G\setminus Z(G)$, the group $C_g^*$ is the unique $\mu$-subgroup of $G$ that contains $g$. The latter follows from our hypothesis and the fact that any $\mu$-subgroup of $G$ containing $g$ is necessarily a $\mu$-subgroup of $C_G(g)$.
\end{proof}

We observe that the condition of Proposition~\ref{prop:com} is satisfied if, for instance, $C_G(g)$ is cyclic for every $g\in G\setminus Z(G)$. The groups satisfying the latter are fully characterized in~\cite{cent}. 

\subsection{Flower groups under the power map}

In this part we prove that the functional graph induced by the power map on flower groups depends only on the type of the group and provide a description of the structure of this graph.

\begin{definition}
Let $G$ be a flower group of type $(c_0;c_1,\ldots,c_k)$ and petals $C_1,\ldots,C_k$. A compatible system of generators for $G$ is $(g_1,\ldots,g_k)$ where $g_i$ is a generator of $C_i$ for $1\leq i \leq k$ and $g_i^{c_i/c_0}=g_j^{c_j/c_0}$ for $1\leq i <j\leq k$.
\end{definition}
We obtain the following result.
%
%

\begin{lemma}\label{lem:csg}
Every flower group has a compatible system of generators.
\end{lemma}

\begin{proof}
Let $G$ be a flower group with pistil $C_0$ and petals $C_1,\ldots,C_k$ and let $c_i:=|C_i|$ for $1\leq i \leq k$. Let $g_1$ be a generator of $C_1$ and consider generators $h_i$ of $C_i$ for $2\leq i \leq k$. Since $h_i^{c_i/c_0}$ and $g_1^{c_1/c_0}$ are both generators of $C_0$ we have that $h_i^{c_i f/c_0}=g_1^{c_1/c_0}$ for some integer $f$ with $\gcd(f,c_0)=1$. Consider a positive integer $f_i$ such that $f_i\equiv f\pmod{c_0}$ and $\gcd(f_i,|G|)=1$ (for example, using Dirichlet's Theorem on arithmetic progressions we can take a prime $f_i$ such that $f_i\equiv f \pmod{c_0}$ and $f_i>|G|$). If we set $g_i=h_i^{f_i}$ for $2\leq i \leq k$, we have that $g_i$ is a generator of $C_i$ satistying $g_i^{c_i/c_0}=h_i^{c_if_i/c_0}=g_1^{c_1/c_0}$.
\end{proof}

The following proposition entails that the isomorphism class of the functional graph induced by power maps on flower groups, depend only on their type.
\begin{proposition} \label{prop:type_isomorphism}
Let $G$ and $H$ be two flower groups of the same type, then there is a graph isomorfism $\psi:\mathcal{G}(\varphi_t/G) \to \mathcal{G}(\varphi_t/H)$ such that $\psi(1_G)=\psi(1_{H})$. 
\end{proposition}

\begin{proof}
Let $C_0$ and $C_0'$ be the pistils of $G$ and $H$, respectively. Let $S=\{C_1,\ldots,C_k\}$ and $S'=\{C_1',\ldots,C_k'\}$ be the set of petal of $G$ and $H$, respectively. Since $G$ and $H$ are flower groups of the same type, with no loss of generality, we may assume that $c_i:=|C_i|=|C_i'|$ for $0\leq i \leq k$. By Lemma \ref{lem:csg} we consider a compatible system of generators $(g_1,\ldots,g_k)$ and $(h_1,\ldots,h_k)$ for $G$ and $H$, respectively. We note that $g_0:=g_1^{c_1/c_0}$ and $h_0:=h_1^{c_1/c_0}$ are generators of $C_0$ and $C_0'$, respectively. For each $i$ with $1\leq i \leq k$ we consider the group isomorphism $\psi_i:C_i\to C_i'$ such that $\psi(g_i)=h_i$ and define $\psi:G\to H$ such that $\psi|_{C_i}=\psi_i$. To prove that this map is well defined it suffices to prove that the maps $\psi_i$ and $\psi_j$ coincides in $C_0$ for $\leq i <j \leq k$. This last assertion follows from the fact that $\psi_i(g_0)=\psi_i(g_i^{c_i/c_0})=h_i^{c_i/c_0}=h_0$ for $1\leq i \leq k$ and that $g_0$ is a generator of $C_0$. It is clear that $\psi$ is a bijection and since $\varphi_t \psi = \psi \varphi_t$ in each petal $C_i$, this also happens globally and $\psi$ induces an isomorphism between the functional graphs $\mathcal{G}(\varphi_t/G)$ and $\mathcal{G}(\varphi_t/H)$.
\end{proof}

Next we provide a description of the functional graph induced by the power map on flower groups. We start with the following important auxiliary result.

\begin{lemma}\label{lem:stable}
Let $G$ be a flower group with petals $C_1, \ldots, C_k$ and pistil $C_0$. Then for any $1\le i\le k$, $\varphi_t^{-1}(C_i\setminus C_0)\subseteq C_i\setminus C_0$. In particular, for an element $g\in C_i\setminus C_0$ that is $\varphi_t$-periodic, we have that the tree attached to $g$ in $\mathcal G(\varphi_t/G)$ is isomorphic to the tree attached to $g$ in $\mathcal G(\varphi_t/C_i)$ and the cycle containing $g$ in $\mathcal G(\varphi_t/G)$ comprises vertices from $C_i\setminus C_0$.  
\end{lemma}

\begin{proof}
Suppose that $h^t\in C_i\setminus C_0$. Hence $h\not\in C_0$ and then $h\in C_j\setminus C_0$ for some $j$ with $1\leq j \leq k$. Since $C_i \cap C_j \setminus C_0 \neq \emptyset$, we conclude that $j=i$. The remaining statement follows directly by the inclusion $\varphi_t^{-1}(C_i\setminus C_0)\subseteq C_i\setminus C_0$.
\end{proof}


We obtain the following result.

\begin{theorem}\label{thm:main}
Let $t$ be a positive integer and $G$ be a flower group of type $(c_0; c_1,\ldots, c_k)$. Set $c_i=\nu_i\cdot \omega_i$ in a way that $\omega_i$ is the greatest divisor of $c_i$ that is relatively prime with $t$. Then the functional graph $\mathcal G(\varphi_t/G)$ of the map $\varphi_t:G\to G$ with $g\mapsto g^t$ is isomorphic to

$$\left(\bigoplus_{i=1}^k\bigoplus_{d_i|\omega_i\atop d_i\nmid\omega_0}\frac{\varphi(d_i)}{\ord_{ d_i}(t)}\times \cyc\left(\ord_{ d_i}(t), \mathcal T_{\gcd_t(\nu_i)}\right)\right)\oplus \left(\bigoplus_{d_0|\omega_0}\frac{\varphi(d_0)}{\ord_{ d_0}(t)}\times \cyc\left(\ord_{ d_0}(t), \mathcal T_{t}(G)\right)\right),$$
where $\mathcal T_{t}(G)$ is the \mqureshi{central tree which} has $\sum_{i=1}^k\nu_i-(k-1)\nu_0$ nodes.

%
%
%
\end{theorem}

\begin{proof}
Let $C_0$ be the pistil of $G$ and $C_1,\ldots,C_k$ be the petals such that $|C_i|=c_i$ for $1\leq i \leq k$. Fix $1\le i\le k$ and let $P_i$ be the set of $\varphi_t$-periodic elements of $C_i\setminus C_0$. From Proposition~\ref{prop:cyc}, $P_i$ has $\omega_i-\omega_0$ elements, corresponding to the cycle decomposition $\bigoplus_{d_i|\omega_i\atop d_i\nmid\omega_0}\frac{\varphi(d_i)}{\ord_{ d_i}(t)}\times \cyc(\ord_{ d_i}(t))$ in $\mathcal G(\varphi_t/G)$. From Proposition \ref{prop:cyc} and  Lemma~\ref{lem:stable}, to each element of $P_i$ is attached a rooted tree isomorphic to $\mathcal T_{\gcd_t(\nu_i)}$; this yields the first component in our statement. It remains to consider the set $P_0$ of $\varphi_t$-periodic elements of $C_0$. From Lemma~\ref{prop:cyc}, $P_0$ has $\omega_0$ elements, corresponding to the cycle decomposition $\bigoplus_{d_i|\omega_0}\frac{\varphi(d_0)}{\ord_{ d_0}(t)}\times \cyc(\ord_{ d_0}(t))$ in $\mathcal G(\varphi_t/G)$.
From Proposition~\ref{prop:properties}, $C_0$ is in the center of $G$ and then, by Proposition~\ref{lem:centre}, the trees attached to the elements of $P_0$ are all isomorphic to $\mqureshi{\mathcal{T}_t(G)}$. The statement about the number of nodes in such tree follows by a counting argument.

\end{proof}

\begin{remark}\label{RemarkNumberOfTrees}
From the above theorem we have that if $G$ is a flower group of type $(c_0;c_1,\ldots,c_k)$ and $t$ is a positive integer then the functional graph $\mathcal{G}(\varphi_t/G)$ has at most $k+1$ non isomorphic trees. 
\end{remark}

\subsection{The rooted tree $\mathcal T_{t}(G)$}

Theorem~\ref{thm:main} provided a description of the functional graph of the power map $\varphi_t$ on a flower group $G$. This description is complete except for the term $\mathcal T_{t}(G)$ which is the rooted tree associated with $1\in G$. In this part we describe it and provide some relations which under some mild conditions allows us to express it in terms of elementary trees.\\

By proposition \ref{prop:type_isomorphism}, the tree $\mathcal T_{t}(G)$ only depends on the type of the flower group $G$. If $G$ has type $(c_0; c_1,\ldots,c_k)$ we denote $\mathcal{T}_{t}(c_0;c_1,\ldots,c_k):=\mathcal T_{t}(G)$.\\

It is convenient to extend the definition of $\mathcal{T}_{t}(c_0;c_1,\ldots,c_k)$ for positive integers $c_0,c_1,\ldots,c_k$ with $c_0\mid c_i$, $1\leq i \leq k$, even if there is no flower group of type $(c_0;c_1,\ldots,c_k)$. We consider set $V= \bigcup_{i=0}^{k}\{i\}\times \Z_{c_i}$. Note that if $x\in \Z_{c_i}$ for some $i$, $1\leq i \leq k$, verifies $x\equiv 0 \pmod{c_i/c_0}$ then there is a unique $x'\in \Z_{c_0}$ such that $x\equiv x'c_i/c_0 \pmod{c_i}$ and we define the identification function $f:V\to V$ given by $$f(i,x)=\left\{ \begin{array}{ll} (i,x) & \textrm{if } i=0 \textrm{ or } x\not\equiv 0 \pmod{\frac{c_i}{c_0}}; \\
(0,x') & \textrm{if } i\neq 0 \textrm{ and } x\equiv x'\cdot \frac{c_i}{c_0} \pmod{c_i}.
\end{array} \right.$$
From the latter, we have the following definition.

\begin{definition}
The pseudo-flower group $F=F(c_0;c_1,\ldots,c_k)$ is the quotient set $V/f$ i.e. its elements are of the form $\overline{(i,x)}:=\{(j,y)\in V: f(j,y)=f(i,x)\}$. The set $V_0:=\{\overline{(i,x)}: x\in \Z_{c_0}\}$ is the pistil of $F$ and the sets $V_i:=\{\overline{(i,x)}: x\in \Z_{c_i}\}$ with $1\leq i \leq k$ are the petals of $F$.
\end{definition}

Note that, for each $i$, $0\leq i \leq k$, the sets $V_i$ have a natural structure of cyclic groups given by $\overline{(i,x)} + \overline{(i,y)} := \overline{(i,x+y)}$ and that this sum is well defined in $F$. Indeed, if $f(i,x_1)=f(j,x_2)=(0,x')$ and $f(i,y_1)=f(j,y_2)=(0,y')$ with $0\leq j < i \leq k$ we have that $x\equiv x'\cdot \frac{c_i}{c_0}\pmod{c_i}$, $y\equiv y'\cdot \frac{c_i}{c_0}\pmod{c_i}$ and $x+y\equiv (x'+y')\cdot \frac{c_i}{c_0}\pmod{c_i}$, thus $f(i,x_1+y_1)=(0,x'+y')$. In a similar way we have that $f(j,x_2+y_2)=(0,x'+y')$ (if $j=0$ it follows directly from the fact that $x_2=x'$ and $y_2=y'$). Then, $f(i,x_1+y_1)=f(j,x_2+y_2)$ for $0\leq j < i \leq k$. It is clear that each $V_i$ is a cyclic group (isomorphic to $\Z_{c_i}$) for $0\leq i \leq k$ and they satisfy $\bigcup_{i=1}^{k}V_i = F$ and $V_i \cap V_j = V_0$ for $1\leq i <j\leq k$ but $F$ is not a genuine flower group since the sum is not defined for every pair of elements of $F$. Nevertheless, we can define the power map $\varphi_t: F \to F$ given by  $\varphi_t(x)=tx$. The group $V_0$ is called the pistil of $F$ and the groups $V_1,\ldots, V_k$ are called the petals of $F$.\\

It is \mqureshi{straightforward to prove} that if $G$ is a flower group of type $(c_0;c_1,\ldots,c_k)$ and $(g_1,\ldots,g_k)$ is a compatible system of generators for $G$, the map $\psi:G \to F(c_0;c_1,\ldots,c_k)$ given by $g_i^j \mapsto \overline{(i,j)}$ for $1\leq i \leq k$ and $j\geq 0$ is well defined (using $\overline{(i,c_i/c_0)}=\overline{(0,1)}$ for $1\leq i \leq k$), bijective and satisfies $\varphi_t \psi = \psi \varphi_t$. Thus, the induced map $\tilde{\psi}: \mathcal{G}(\varphi_t/G) \to \mathcal{G}(\varphi_t/F(c_0;c_1,\ldots,c_k))$ is a graph isomorphism. Since $\psi(1)=\overline{(0,0)}$, the tree $\mathcal{T}_t(c_0;c_1,\ldots,c_k)$ is isomorphic to the tree attached to $\overline{(0,0)}$ in $\mathcal{G}(\varphi_t/F(c_0;c_1,\ldots,c_k))$.\\

Now we have defined the rooted tree $\mathcal{T}_t(c_0;c_1,\ldots,c_k)$ for every $c_0,c_1,\ldots,c_k$ satisfying $c_0\mid c_i$ even if there is no flower group of type $(c_0;c_1,\ldots,c_k)$. In the case that there is a flower group $G$ of type $(c_0;c_1,\ldots,c_k)$, this tree coincides with \mqureshi{the central tree $\mathcal T_{t}(G)$} in the functional graph $\mathcal{G}(\varphi_t/G)$. \\

The following proposition brings us a method which allows to simplify the tree structure of $\mathcal{T}_t(c_0;c_1,\ldots,c_k)$.

\begin{theorem}\label{thm:tree}
Let $c_0,c_1,\ldots, c_k$ be positive integers with $k\geq 2$ and $c_0\mid c_i$ for $1\leq i \leq k$. Then, the following holds:
\begin{enumerate}[i)]
\item $\mathcal{T}_t(c_0;c_1,\ldots,c_k)=\mathcal{T}_t(c_0;c_{\theta(1)},\ldots,c_{\theta(k)})$ for every permutation $\theta$ of the set $\{1, 2, \ldots, k\}$.
\item If $\gcd(c_0, t)=1$, then $\mathcal T_{t}(c_0;c_1,\ldots,c_k)=\sum_{i=1}^k\mathcal T_{\gcd_t(\nu_i)}$.
\item If $\gcd(t,\frac{c_k}{c_0})=1$ then $\mathcal{T}_t(c_0;c_1,\ldots,c_{k-1},c_k)=\mathcal{T}_t(c_0;c_1,\ldots,c_{k-1})$.
\item If $c_k \mid t$ then $\mathcal{T}_t(c_0;c_1,\ldots,c_{k-1},c_k)=\mathcal{T}_t(c_0;c_1,\ldots,c_{k-1})+\langle (c_k-c_0)\times \bullet \rangle$.
\item $\mathcal{T}_t(c_0;c_1)=\mathcal{T}_{\gcd_t(c_1)}$. 
\end{enumerate}
\end{theorem}

\begin{proof}
Consider the pseudo-flower group $F=F(c_0;c_1,\ldots,c_k)$ with pistil $V_0$ and petals $V_1,\ldots,V_k$ and the power map $\varphi_t:F \to F$. We split the proof into cases:

\begin{enumerate}[i)]
\item Item i) follows directly from the definition of $F$. 
\item We observe that, if $\gcd(c_0, t)=1$, then every element of $V_0$ is $\varphi_t$-periodic. In particular, the tree attached to $\overline{(0,0)}\in F$ contains, in addition to the root $\overline{(0,0)}$, only elements of the set $F\setminus V_0=\bigcup_{i=1}^k V_i\setminus V_0$. Such tree is the union of the rooted trees associated with $\overline{(0,0)}\in F$ in each of the functional graphs $\mathcal G(\varphi_t/V_i)$, that is, $\sum_{i=1}^k\mathcal T_{\gcd_t(c_i)}$. 
\item We first note that $\gcd(t,c_k/c_0)=1$ implies $\varphi_t(V_k\setminus V_0) \subseteq V_k\setminus V_0$. Indeed, if $\overline{(k,x)} \in V_k\setminus V_0$ then $\frac{c_k}{c_0} \nmid x$ and, since $\gcd\left(t,\frac{c_k}{c_0}\right)=1$, we also have $\frac{c_k}{c_0} \nmid tx$. Hence, $\varphi_t(\overline{(k,x)}) = \overline{(k,tx)} \in V_k\setminus V_0$. The inclusion $\varphi_t(V_k\setminus V_0) \subseteq V_k\setminus V_0$ implies that there are no points of the tree attached to $\overline{(0,0)}$ in $V_i\setminus V_0$ and it is equal to $\mathcal{T}_t(c_0;c_1,\ldots,c_{k-1})$.
\item We observe that if $c_k\mid t$ then $\varphi_t(V_k)=\{\overline{(0,0)}\}$. Set $F^* = (F\setminus V_k) \cup V_0$. The tree attached to $\overline{(0,0)}$ in $\mathcal{G}(\varphi_t/F)$ can be obtained as the union of the tree attached to $\overline{(0,0)}$ in $\mathcal{G}(\varphi_t/F^*)$ (which equals $\mathcal{T}_t(c_0;c_1,\ldots,c_{k-1})$) and the tree $\langle (c_k-c_0)\times \bullet \rangle$ (corresponding to the $c_k-c_0$ points of $V_k\setminus V_0$ mapping to $\overline{(0,0)}$ by $\varphi_t$), both trees with the same root, then $\mathcal{T}_t(c_0;c_1,\ldots,c_{k-1},c_k)=\mathcal{T}_t(c_0;c_1,\ldots,c_{k-1})+\langle (c_k-c_0)\times \bullet \rangle$.
\item We note that $F(c_0;c_1)=V_1$ is a cyclic group of order $c_1$ and then the equality $\mathcal{T}_t(c_0;c_1)=  \mathcal{T}_{\gcd_t(c_1)}$ follows from Proposition \ref{prop:cyc}.
\end{enumerate}

\end{proof}

\section{Applications}\label{Section:applications}

In this section we provide some applications of Theorem~\ref{thm:main} to the explicit description of the functional graph of the power map on certain classes of finite groups. For clarity and organization, we consider them separately.

\subsection{Generalized quaternions}
In~\cite{a2} the authors obtain an implicit description of the digraph associated to power maps over generalized quaternion groups, such as distribution of indegrees and cycle lengths. Here we a provide a more explicit description of these digraphs.
We start by showing that such groups are flower groups.
\begin{lemma}\label{lem:quat}
For $n\ge 2$, the generalized quaternion $$Q_{4n}=\langle a, b\,|\, a^{2n}=1, \, a^n=b^2\, , bab^{-1}=a^{-1}\rangle,$$ of order $4n$ is a flower group with pistil $C_0=\{1, a^n\}$. Moreover, the set of petals of $Q_{4n}$ comprises $n$ cyclic groups of order $4$ and one cyclic group of order $2n$.
\end{lemma}

\begin{proof}
We observe that every element of $Q_{4n}$ is written uniquely as $a^ib^j$, where $j=0, 1$ and $0\le i\le 2n-1$. Therefore, the $\mu$-subgroups of $Q_{4n}$ are the groups $C_{i}=\langle a^ib\rangle=\{1, a^ib, a^n, a^{n+i}b\}$ with $1\le i\le n$, each of order $4$ and $C_{n+1}=\langle a \rangle$, which has order $2n$. From this fact the result follows directly.
\end{proof}

We obtain the following corollary.


\begin{corollary}\label{cor:qua}
Fix $t$ an integer, let $Q_{4n}$ be as in Lemma~\ref{lem:quat} and set $2n=\nu\cdot \omega$ in a way that $\omega$ is the greatest divisor of $2n$ that is relatively prime with $t$. Let $\alpha$ be the integer such that $n/2^\alpha$ is odd. Then the functional graph $\mathcal G(\varphi_t/Q_{4n})$ of the map $\varphi_t:Q_{4n}\to Q_{4n}$ with $g\mapsto g^t$ is isomorphic to one of the following graphs:
$$\left(\bigoplus_{d|\omega\atop d> 2}\frac{\varphi(d)}{\ord_{d}(t)}\times \cyc(\ord_{d}(t), \mathcal T_{\gcd_t(\nu)})\right)\oplus \left( k_tn\times \cyc \left(\frac{2}{k_t}\right)\right) \oplus \left(2\times   \{\mathcal T_{\gcd_t(\nu)}\}\right),$$
if $t$ is odd, where $k_t=2$ if $t\equiv 1\pmod 4$ and $k_t=1$ if $t\equiv 3\pmod 4$, or
$$\left(\bigoplus_{d|\omega\atop d> 1}\frac{\varphi(d)}{\ord_{d}(t)}\times \cyc(\ord_{d}(t), \mathcal T_{\gcd_t(\nu)})\right) \oplus \{T_0\}, $$
if $t$ is even, where $T_0 = \left\{ \begin{array}{ll}
T_{\gcd_t(\nu)} + \langle 2n \times \bullet \rangle & \textrm{if } t\equiv 0 \pmod{4}; \\
T_{\gcd_t(\nu)} +_\alpha \langle 2n \times \bullet \rangle & \textrm{if } t\equiv 2 \pmod{4}.
\end{array} \right.$
\end{corollary}

\begin{proof}
By Lemma~\ref{lem:quat}, the type of $Q_{4n}$ is $(c_0;c_1,\ldots,c_{n+1})$ with $c_0=2$, $c_i=4$ for $1\leq i \leq n$ and $c_{n+1}=2n$. Let $C_0$ be the pistil and $C_1,\ldots,C_{n+1}$ be the petals of $Q_{4n}$ with $|C_i|=c_i$ for $1\le i\le n+1$. Set $c_i=\nu_i\cdot \omega_i$ such that $\omega_i$ is the greatest divisor of $c_i$ that is relatively prime with $t$. In all the cases we apply Theorem~\ref{thm:main} to obtain the graph structure of $\mathcal G(\varphi_t/Q_{4n})$ except for the tree $\mathcal{T}_t(Q_{4n})$ which equals $T_0$ when $t\equiv 0 \pmod{2}$. We split the proof into cases.

\begin{enumerate}[a)]
\item If $t\equiv 1 \pmod{2}$ we have $\gcd(t,\frac{c_i}{c_0})=1$ for $1\leq i \leq n$ and by Theorem \ref{thm:tree} (repeated application of i. and iii. and using v. to finish) we obtain $\mathcal{T}_t(Q_{4n})=T_{\gcd_t(\nu)}$.

\item  If $t\equiv 0 \pmod{4}$ we have $c_i\mid t$ for $1\leq i \leq n$ and by Theorem \ref{thm:tree} (repeated application of i. and iv. and using v. to finish) we obtain $T_0= T_{\gcd_t(\nu)} + \langle 2n \times \bullet \rangle$. 

\item If $t\equiv 2 \pmod{4}$ and $1\leq i \leq n$, the power map $\varphi_t$ maps the two points of $C_i\setminus C_0$ in the unique point of order $2$ in $C_0$ and the hanging tree of $1$ in $\mathcal{G}(\varphi_t/C_{n+1})$ is $T_{\gcd_t(\nu)}$. Then, the tree $T_0$ can be obtained from $T_{\gcd_t(\nu)}$ by replacing the subtree $T$ whose node corresponds to the point of order $2$ in $C_{n+1}$ with $T+\langle 2n \times \bullet \rangle$, that is, $T_0 = T_{\gcd_t(\nu)} +_\alpha \langle 2n \times \bullet \rangle$ where $\alpha$ is the depth $T$. To conclude we observe that $\alpha$ equals the greatest integer $k$ for which there exists $x\in \Z$ such that $t^k x\equiv n \pmod{2n}$. This last equation has a solution if and only if $\gcd(t^k,2n)\mid n$. If  $e_2(m)$ denotes the exponent of $2$ in the prime decomposition of $m$, the latter is equivalent to $\min\{k\cdot e_2(t), 1+e_2(n)\}\leq e_2(n)$, i.e., $k\cdot e_2(t)\leq e_2(n)$. Since $e_2(t)=1$ we have $\alpha=e_2(n)$ as desired.

\end{enumerate}
\end{proof}

We provide two numerical  examples, showing the applicability of Corollary~\ref{cor:qua}.

\begin{example}
Consider the group $Q_{24} (n=6)$ and the map $\varphi_3:Q_{24}\to Q_{24}$ with $\varphi_3(g)=g^3$. From Corollary~\ref{cor:qua}, we obtain that $\mathcal G(\varphi_{3}/Q_{24})$ is isomorphic to
$\cyc(2, \mathcal T_{(3)})\oplus (2\times \cyc(1, \mathcal T_{(3)}))\oplus (6\times \cyc(2))$; see Figure~\ref{fig:Q24}.


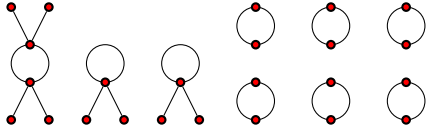
\begin{figure}[H]
\begin{center}
\begin{tikzpicture}
  [scale=1, place/.style={circle,draw=black,thick,fill=red,
                 inner sep=0pt,minimum size=1mm}]
\draw (0,0.25) circle (2.5mm);
  \node (1) at ( 0,0) [place] {};
  \node (2) at (-0.25 ,-0.5) [place] {};
  \node (3) at (0.25, -0.5) [place] {};
  \node (4) at ( 0,0.5) [place] {};
  \node (5) at (-0.25 ,1) [place] {};
  \node (6) at (0.25, 1) [place] {};
\foreach \x/\y in {1/2, 1/3, 4/5, 4/6}
	\draw (\x) to (\y);

\begin{scope}[xshift=1cm, yshift=0]
\draw (0,0.25) circle (2.5mm);
  \node (1) at ( 0,0) [place] {};
  \node (2) at (-0.25 ,-0.5) [place] {};
  \node (3) at (0.25, -0.5) [place] {};
\foreach \x/\y in {1/2, 1/3}
	\draw (\x) to (\y);
\end{scope}

\begin{scope}[xshift=2cm, yshift=0]
\draw (0,0.25) circle (2.5mm);
  \node (1) at ( 0,0) [place] {};
  \node (2) at (-0.25 ,-0.5) [place] {};
  \node (3) at (0.25, -0.5) [place] {};
\foreach \x/\y in {1/2, 1/3}
	\draw (\x) to (\y);
\end{scope}

\begin{scope}[xshift=3cm, yshift=0.5cm]
\draw (0,0.25) circle (2.5mm);
  \node (1) at ( 0,0) [place] {};
  \node (2) at (0,0.5) [place] {};
\end{scope}

\begin{scope}[xshift=4cm, yshift=0.5cm]
\draw (0,0.25) circle (2.5mm);
  \node (1) at ( 0,0) [place] {};
  \node (2) at (0,0.5) [place] {};
\end{scope}

\begin{scope}[xshift=5cm, yshift=0.5cm]
\draw (0,0.25) circle (2.5mm);
  \node (1) at ( 0,0) [place] {};
  \node (2) at (0,0.5) [place] {};
\end{scope}

\begin{scope}[xshift=3cm, yshift=-0.5cm]
\draw (0,0.25) circle (2.5mm);
  \node (1) at ( 0,0) [place] {};
  \node (2) at (0,0.5) [place] {};
\end{scope}

\begin{scope}[xshift=4cm, yshift=-0.5cm]
\draw (0,0.25) circle (2.5mm);
  \node (1) at ( 0,0) [place] {};
  \node (2) at (0,0.5) [place] {};
\end{scope}

\begin{scope}[xshift=5cm, yshift=-0.5cm]
\draw (0,0.25) circle (2.5mm);
  \node (1) at ( 0,0) [place] {};
  \node (2) at (0,0.5) [place] {};
\end{scope}

\end{tikzpicture}
\end{center}
\caption{The graph $\cyc(2, \mathcal T_{(3)})\oplus (2\times \cyc(1, \mathcal T_{(3)}))\oplus (6\times \cyc(2))$.}
\label{fig:Q24}
\end{figure}


\end{example}

\begin{example}
Consider the group $Q_{48} (n=12)$ and the map $\varphi_{10}:Q_{48}\to Q_{48}$ with $\varphi_{10}(g)=g^{10}$. From Corollary~\ref{cor:qua}, we obtain that $\mathcal G(\varphi_{10}/Q_{48})$ is isomorphic to $2\times \{\mathcal{T}_{(2,2,2)}\} \oplus \{\mathcal{T}_{(2,2,2)} +_{2} \langle 24 \times \bullet \rangle  \}$; see Figure~\ref{fig:Q48}.


\begin{figure}[H]
\begin{center}
\begin{tikzpicture}
  [scale=1, place/.style={circle,draw=black,thick,fill=red,
                 inner sep=0pt,minimum size=1mm}]
\draw (0,0.25) circle (2.5mm);
  \node (1) at (0,0) [place] {};
  \node (2) at (0,-0.5) [place] {};
  \node (3) at (0.35, -1) [place] {};
  \node (4) at (-0.35,-1) [place] {};
  \node (5) at (0.1,-1.5) [place] {};
  \node (6) at (0.6,-1.5) [place] {};
  \node (7) at (-0.1,-1.5) [place] {};
  \node (8) at (-0.6,-1.5) [place] {};
\foreach \x/\y in {1/2, 2/3, 4/8, 2/4, 3/5, 3/6, 4/7}
	\draw (\x) to (\y); 

\begin{scope}[xshift=2cm, yshift=0]
\draw (0,0.25) circle (2.5mm);
  \node (1) at (0,0) [place] {};
  \node (2) at (0,-0.5) [place] {};
  \node (3) at (0.35, -1) [place] {};
  \node (4) at (-0.35,-1) [place] {};
  \node (5) at (0.1,-1.5) [place] {};
  \node (6) at (0.6,-1.5) [place] {};
  \node (7) at (-0.1,-1.5) [place] {};
  \node (8) at (-0.6,-1.5) [place] {};
\foreach \x/\y in {1/2, 2/3, 4/8, 2/4, 3/5, 3/6, 4/7}
	\draw (\x) to (\y);
\end{scope}     
           
\begin{scope}[xshift=4cm, yshift=0]
\draw (0,0.25) circle (2.5mm);
  \node (1) at (0,0) [place] {};
  \node (2) at (0,-0.5) [place] {};
  \node (3) at (0.35, -1) [place] {};
  \node (4) at (-0.35,-1) [place] {};
  \node (5) at (0.1,-1.5) [place] {};
  \node (6) at (0.6,-1.5) [place] {};
  \node (7) at (-0.1,-1.5) [place] {};
  \node (8) at (-0.6,-1.5) [place] {};
  \node (10) at (0.8, -1) [place] {};
  \node (11) at (1, -1) [place] {};
  \node (ptos) at (1.25,-1) {...};
  \node (12) at (1.5,-1) [place] {};
  \node (underbrace) at (1.15,-1.3) {\scriptsize $\underbrace{\hspace{0.8cm}}_{24}$};
\foreach \x/\y in {1/2, 2/3, 4/8, 2/4, 3/5, 3/6, 4/7}
	\draw (\x) to (\y);
  \draw (2) to [bend left=20] (10);
  \draw (2) to [bend left=20] (11);
  \draw (2) to [bend left=20] (12);
\end{scope}     
              
\end{tikzpicture}
\end{center}
\caption{The graph $2\times \{\mathcal{T}_{(2,2,2)}\} \oplus \{\mathcal{T}_{(2,2,2)} +_{2} \langle 24 \times \bullet \rangle  \}$.}
\label{fig:Q48}
\end{figure}
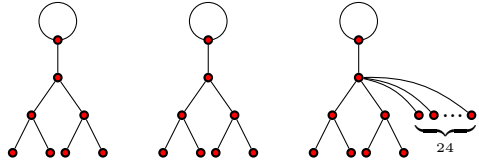


\end{example}

\subsection{Semidirect products of cyclic groups}

Let $\C_n=\langle a \rangle$ and $\C_m= \langle b \rangle$ be two cyclic groups. Every homomorphism $\phi:\C_m \to \operatorname{Aut}(\C_n)$ is defined by $\phi(b)(a)=a^s$ for some integer $s$ with $s^n\equiv 1 \pmod{m}$ and its associated semidirect product is denoted by $\C_n\rtimes_s \mathcal C_m$. 

In~\cite{deng}, the authors explored the digraph associated to the power map on such groups in the case where $n$ is a prime number. The following lemma provides a class of such groups that are also flower groups.

\begin{lemma}\label{lem:semi}
Let $m, n, s>1$ be positive integers such that \mqureshi{$\frac{s^m-1}{s-1}\equiv 0 \pmod{n}$} and $\gcd\left(n, \frac{s^j-1}{s-1}\right)=1$ for $1\le j<m$. Then the semi-direct product $$\C_n\rtimes_s \mathcal C_m=\langle a, b\,|\, b^n=a^m=1\, , aba^{-1}=b^s \rangle, $$ is a flower group of order $mn$ with pistil $C_0=\{1\}$. Moreover, the petals of $\C_n\rtimes_s \mathcal C_m$ comprises $n$ cyclic groups of order $m$ and one cyclic group of order $n$.
\end{lemma}

\begin{proof}

We observe that every element of $G:=\C_n\rtimes_s \mathcal C_m$ is written uniquely as $b^ia^j$ with $0\le i<n$ and $0\le j<m$. It is direct to verify that, for any $0\le i, i_0<n$ and $0\le j, j_0<m$ we have that $(b^ia^j)\cdot (b^{i_0}a^{j_0})=b^{i+i_0\cdot s^j}a^{j+j_0}$. In particular, we obtain that 
\begin{equation}\label{eq:aux-dic}(b^ia^j)^t=b^{i\cdot \frac{s^{jt}-1}{s^j-1}}a^{jt},\, t>0.\end{equation} 

We claim that for every element $g=b^{i} a^{j}\in G$ with $0\leq i <n$, $1\leq j<m$ there is an integer $i_0$ such that $(b^{i_0} a)^{j} = b^{i}a^{j}$. Indeed, by Equation \ref{eq:aux-dic}, it suffices to take $i_0$ such that $i_0\left( \frac{s^j-1}{s-1} \right) \equiv i \pmod{n}$ and the existence of such $i_0$ is guaranteed by the fact that $\gcd\left(n, \frac{s^j-1}{s-1}\right)=1$ for every $1\le j<m$. Thus, the $\mu$-subgroups are $\langle b \rangle$ (of order $n$) and $\langle b^i a \rangle$ with $0\leq i <n$. \mqureshi{The order of $b^i a$ is $m$ because $(b^i a)^{t}=b^{i\cdot \frac{s^{t}-1}{s-1}}a^{t}=1$ implies $t\equiv 0 \pmod{m}$ and for $t=m$ we have $(b^i a)^{m} = b^{i\cdot \frac{s^{m}-1}{s-1}}=1$ because $\frac{s^m-1}{s-1}\equiv 0 \pmod{n}$.} Since the union \mqureshi{of the $\mu$-subgroups} covers $G$ and $(n-1)+(m-1)n=mn-1= |G|-1$ we conclude that they are the $\mu$-subgrups with pairwise trivial intersection which implies that $G$ is a flower group. 
\end{proof}

\begin{remark}
The conditions on $s, m$ and $n$ in Lemma~\ref{lem:semi} are fulfilled if, for instance, the order of $s$ modulo $p$ equals $m$ for every prime divisor $p$ of $n$.
\end{remark}

\begin{remark}
For every positive integer $n\ge 2$, the pair $(s, m)=(n-1, 2)$ satisfies the conditions of Lemma~\ref{lem:semi}. In this case, the corresponding group is just the dihedral group $\mathcal D_{2n}$ of order $2n$.
\end{remark}

The following corollary is a direct application of Lemma~\ref{lem:semi} and Theorem~\ref{thm:main}.

\begin{corollary}\label{cor:semi}
Fix $t$ a positive integer and let $m, n, s>1$ and  $\C_n\rtimes_s \mathcal C_m$ be as in Lemma~\ref{lem:semi}. Write $n=\nu_1\omega_1$ in a way that $\omega_1$ is the greatest divisor of $n$ that is relatively prime with $t$ and write $m=\nu_2\omega_2$ in the same way. Then, for $n_1=1$ and $n_2=n$, the functional graph $\mathcal G(\varphi_t/\C_n\rtimes_s \mathcal C_m)$ of the map $\varphi_t$ over $\C_n\rtimes_s \C_m$ is isomorphic to
$$\left(\bigoplus_{i=1, 2}\bigoplus_{d_i|\omega_i\atop d_i\ne 1}\frac{n_i\cdot \varphi(d_i)}{\ord_{ d_i}(t)}\times \cyc(\ord_{ d_i}(t), \mathcal T_{\gcd_t(\nu_i)})\right)\oplus \{ \mathcal T_{\gcd_t(\nu_1)}+n\cdot \mathcal T_{\gcd_t(\nu_2)}\}.$$
\end{corollary}

We provide a numerical example, showing the applicability of Corollary~\ref{cor:semi}.

\begin{example}
\mqureshi{Consider the group $G=\C_{65}\rtimes_{8} \C_{4}$ as in Lemma~\ref{lem:semi} and let $\varphi_{10}:G\to G$ with $\varphi_{10}(g)=g^{10}$. From Corollary~\ref{cor:semi}, we have that $G(\varphi_{10}/G)$ is isomorphic to $ 2\times \cyc(6, \mathcal T_{(5)})\oplus \{ \mathcal{T}_{(5)} + 65\cdot \mathcal{T}_{(2,2)} \}$. Figure~\ref{fig:Semidirect} shows a picture of this graph.}


\begin{figure}[H]
\begin{center}
\begin{tikzpicture}
  [scale=1, place/.style={circle,draw=black,thick,fill=red,
                 inner sep=0pt,minimum size=1mm}]
\draw (0,0) circle (5mm);
  \node (1) at (0,0.5) [place] {};
  \node (2) at (0.1,1) [place] {};
  \node (3) at (0.3,1) [place] {};
  \node (4) at (-0.1,1) [place] {};
  \node (5) at (-0.3,1) [place] {};
\foreach \x/\y in {1/2, 1/3, 1/4, 1/5}
	\draw (\x) to (\y); 

\begin{scope}[rotate=60]
  \node (1) at (0,0.5) [place] {};
  \node (2) at (0.1,1) [place] {};
  \node (3) at (0.3,1) [place] {};
  \node (4) at (-0.1,1) [place] {};
  \node (5) at (-0.3,1) [place] {};
\foreach \x/\y in {1/2, 1/3, 1/4, 1/5}
	\draw (\x) to (\y);   
\end{scope}

\begin{scope}[rotate=120]
  \node (1) at (0,0.5) [place] {};
  \node (2) at (0.1,1) [place] {};
  \node (3) at (0.3,1) [place] {};
  \node (4) at (-0.1,1) [place] {};
  \node (5) at (-0.3,1) [place] {};
\foreach \x/\y in {1/2, 1/3, 1/4, 1/5}
	\draw (\x) to (\y);   
\end{scope}

\begin{scope}[rotate=180]
  \node (1) at (0,0.5) [place] {};
  \node (2) at (0.1,1) [place] {};
  \node (3) at (0.3,1) [place] {};
  \node (4) at (-0.1,1) [place] {};
  \node (5) at (-0.3,1) [place] {};
\foreach \x/\y in {1/2, 1/3, 1/4, 1/5}
	\draw (\x) to (\y);   
\end{scope}

\begin{scope}[rotate=240]
  \node (1) at (0,0.5) [place] {};
  \node (2) at (0.1,1) [place] {};
  \node (3) at (0.3,1) [place] {};
  \node (4) at (-0.1,1) [place] {};
  \node (5) at (-0.3,1) [place] {};
\foreach \x/\y in {1/2, 1/3, 1/4, 1/5}
	\draw (\x) to (\y);   
\end{scope}

\begin{scope}[rotate=300]
  \node (1) at (0,0.5) [place] {};
  \node (2) at (0.1,1) [place] {};
  \node (3) at (0.3,1) [place] {};
  \node (4) at (-0.1,1) [place] {};
  \node (5) at (-0.3,1) [place] {};
\foreach \x/\y in {1/2, 1/3, 1/4, 1/5}
	\draw (\x) to (\y);   
\end{scope}

\begin{scope}[xshift=3cm, yshift=0cm]

\draw (0,0) circle (5mm);
  \node (1) at (0,0.5) [place] {};
  \node (2) at (0.1,1) [place] {};
  \node (3) at (0.3,1) [place] {};
  \node (4) at (-0.1,1) [place] {};
  \node (5) at (-0.3,1) [place] {};
\foreach \x/\y in {1/2, 1/3, 1/4, 1/5}
	\draw (\x) to (\y); 

\begin{scope}[rotate=60]
  \node (1) at (0,0.5) [place] {};
  \node (2) at (0.1,1) [place] {};
  \node (3) at (0.3,1) [place] {};
  \node (4) at (-0.1,1) [place] {};
  \node (5) at (-0.3,1) [place] {};
\foreach \x/\y in {1/2, 1/3, 1/4, 1/5}
	\draw (\x) to (\y);   
\end{scope}

\begin{scope}[rotate=120]
  \node (1) at (0,0.5) [place] {};
  \node (2) at (0.1,1) [place] {};
  \node (3) at (0.3,1) [place] {};
  \node (4) at (-0.1,1) [place] {};
  \node (5) at (-0.3,1) [place] {};
\foreach \x/\y in {1/2, 1/3, 1/4, 1/5}
	\draw (\x) to (\y);   
\end{scope}

\begin{scope}[rotate=180]
  \node (1) at (0,0.5) [place] {};
  \node (2) at (0.1,1) [place] {};
  \node (3) at (0.3,1) [place] {};
  \node (4) at (-0.1,1) [place] {};
  \node (5) at (-0.3,1) [place] {};
\foreach \x/\y in {1/2, 1/3, 1/4, 1/5}
	\draw (\x) to (\y);   
\end{scope}

\begin{scope}[rotate=240]
  \node (1) at (0,0.5) [place] {};
  \node (2) at (0.1,1) [place] {};
  \node (3) at (0.3,1) [place] {};
  \node (4) at (-0.1,1) [place] {};
  \node (5) at (-0.3,1) [place] {};
\foreach \x/\y in {1/2, 1/3, 1/4, 1/5}
	\draw (\x) to (\y);   
\end{scope}

\begin{scope}[rotate=300]
  \node (1) at (0,0.5) [place] {};
  \node (2) at (0.1,1) [place] {};
  \node (3) at (0.3,1) [place] {};
  \node (4) at (-0.1,1) [place] {};
  \node (5) at (-0.3,1) [place] {};
\foreach \x/\y in {1/2, 1/3, 1/4, 1/5}
	\draw (\x) to (\y);   
\end{scope}

\end{scope}

\begin{scope}[xshift=6cm, yshift=0.5cm]
\draw (0,0.25) circle (2.5mm);
  \node (1) at (0,0) [place] {};
  \node (2) at (0.3,-0.5) [place] {};
  \node (3) at (-0.3,-0.5) [place] {};
  \node (21) at (0.1,-1) [place] {};
  \node (22) at (0.5,-1) [place] {};
  \node (4) at (-0.5,-0.5) [place] {};
  \node (5) at (-0.7,-0.5) [place] {};
  \node (6) at (-0.9,-0.5) [place] {};  
  \node (7) at  (0.9,-0.5) [place] {};
  \node (71) at (0.7,-1) [place] {};
  \node (72) at (1.1,-1) [place] {};
  \node (ptos) at (1.5,-0.75) {.....};
  \node (8) at  (2.1,-0.5) [place] {};
  \node (81) at (1.9,-1) [place] {};
  \node (82) at (2.3,-1) [place] {};
  \node (underbrace) at (1.2, -1.3) {\scriptsize $\underbrace{\hspace{2.4cm}}_{65}$};
\foreach \x/\y in {1/2, 1/3, 1/4, 1/5, 1/6, 1/7, 1/8, 2/21, 2/22, 7/71, 7/72, 8/81, 8/82}
	\draw (\x) to (\y);
\end{scope}

\end{tikzpicture}
\end{center}
\caption{The graph $ 2\times \cyc(6, \mathcal T_{(5)})\oplus \{ \mathcal{T}_{(5)} + 65\cdot \mathcal{T}_{(2,2)} \}$.}
\label{fig:Semidirect}
\end{figure}
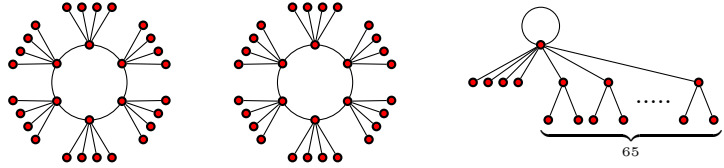


\end{example}

\subsection{The projective general linear group}

Fix $q=p^s$ a prime power and let $\F_q$ be the finite field with $q$ elements. The Projective General Linear group of order 2 over $\F_q$ is the quotient group $\mathbb G_q:=\PGL(2, q)=\frac{\mathrm{GL}(2, q)}{\F_q^*\cdot I}$, where $\mathrm{GL}(2, q)$ is the group of the $2\times 2$ non-singular matrices with entries in $\F_q$ and $I$ is the $2\times 2 $ identity matrix. It is well known that $\mathbb G_q$ has order $q^3-q$. We shall prove that $\mathbb G_q$ is, in fact, a flower group. For this, we need the following machinery.

\begin{definition}
For $A\in \GL(2, q)$ such that $[A]\ne [I]$, $A$ is of {\bf type} $1$ (resp. $2$, $3$ or $4$) if its eigenvalues are distinct and in $\F_q$ (resp. equal and in $\F_q$, symmetric and in $\F_{q^2}\setminus \F_q$ or not symmetric and in $\F_{q^2}\setminus \F_q$). 
\end{definition}

We observe that elements of type $3$ appear only when $q$ is odd. Moreover, the types of $A$ and $\lambda\cdot A$ are the same for any $\lambda\in \F_q^*$. For this reason, we say that $[A]$ is of type $t$ if $A$ is of type $t$. The number of elements of each type and the structure of the centralizers is well known and we display them in Table~\ref{table:1}.

\begin{table}[ht]
\begin{center}
\begin{tabular}{|c|c|c|c|c|}\hline
$p$ & \text{type of} $g$& \text{order of } $g$& $C_{\mathbb G_q}(g)$ &   $\#$ \text{elements} \\ \hline
 &  $1$ & $>2$ &  $\cong  \C_{q-1}$ &  $q(q+1)(q-2)/2$  \\ \cline{2-5}
 $2$ & $2$ & $2$ &  $\cong \mathbb \C_2^s$ &    $q^2-1$\\ \cline{2-5}
 &  $4$ & $>2$  &$\cong \mathbb \C_{q+1}$ &  $q^2(q-1)/2$\\ \hline
 &  $1$ &  $>2$ &  $\cong \mathbb \C_{q-1}$ &   $q(q+1)(q-3)/2$ \\ \cline{2-5}
 &  $1$ & $2$ &$\cong \mathcal{D}_{2(q-1)}$ &   $q(q+1)/2$ \\ \cline{2-5}
 $ >2 $ & $2$ & $p$ & $\cong \mathbb \C_p^s$   &$q^2-1$
 \\ \cline{2-5}
 &  $3$ &  $2$ & $\cong \mathcal{D}_{2(q+1)}$ &  $q(q-1)/2$  \\ \cline{2-5}
 &  $4$ &  $>2$ & $\cong \mathbb \C_{q+1}$&  $q(q-1)^2/2$\\ \hline
\end{tabular}
\caption{Element structure in $\mathbb G_q=\mathrm{PGL}(2, q)$, where $q=p^s$. Here $\C_n$ and $\mathcal{D}_{2n}$ denote the cyclic group or order $n$ and the dihedral group of order $2n$, respectively.}
\label{table:1}
\end{center}
\end{table}

We obtain the following result.

\begin{proposition}\label{prop:pgl}
For any prime power $q=p^s$, the projective general linear group $\mathbb G_q=\PGL(2, q)$ is a flower group with pistil $C_0=\{[I]\}$. Moreover, for $q\ge 3$, the set of petals of $\PGL(2, q)$ comprises $\frac{q(q+1)}{2}$ cyclic groups of order $q-1$, $\frac{q(q-1)}{2}$ cyclic groups of order $q+1$ and $\frac{q^2-1}{p-1}$ cyclic groups of order $p$.
\end{proposition}

\begin{proof}
It is well known that the center of $\mathbb G_q$ equals $\{[I]\}$. By Proposition \ref{prop:com}, it suffices to prove that every non identity element $g\in \mathbb{G}_q$ is contained in a unique $\mu$-subgroup of $C_G(g)$. This is trivially verified if $C_G(g)$ is cyclic or has prime exponent. Otherwise, $C_G(g)$ is a dihedral group, the order of $g$ is $2$ and $q$ is odd (see Table \ref{table:1}). In this case, we observe that $n=q\pm 1$ is even and the center of $\mathcal D_{2n}$ equals $\{1, h^{\frac{n}{2}}\}$ where $h$ is an element of order $n$. Hence $g=h^{\frac{n}{2}}$ and it is direct to verify that such element is contained in a unique $\mu$-subgroup of $\mathcal D_{2n}$, namely the group of order $n$ generated by $h$. Thus $\mathbb G_q$ is a flower group with pistil $\{[I]\}$ whose set of petals comprises the cyclic subgroups of orders $p, q+1$ and $q-1$. A detailed account in Table~\ref{table:1} provides the number of such subgroups, according to their orders.
\end{proof}

The following corollary is an immediate application of Proposition~\ref{prop:pgl} and Theorems~\ref{thm:main} and \ref{thm:tree}.

\begin{corollary}\label{cor:pgl}
Fix $t$ a positive integer, let $q=p^s\ge 4$ be a prime power and $\mathbb G_q=\PGL(2, q)$. Write $q-1=\nu_1\omega_1$ in a way that $\omega_1$ is the greatest divisor of $q-1$ that is relatively prime with $t$ and write $q+1=\nu_2\omega_2$ in the same way. Then, for $d_1=\frac{q(q+1)}{2}$ and $d_2=\frac{q(q-1)}{2}$, the functional graph $\mathcal G(\varphi_t/\mathbb G_q)$ of the map $\varphi_t:\mathbb G_q\to \mathbb G_q$ with $g\mapsto g^t$ is isomorphic to
$$\left(\bigoplus_{i=1, 2}\bigoplus_{d_i|\omega_i\atop d_i\ne 1}\frac{d_i\cdot \varphi(d_i)}{\ord_{ d_i}(t)}\times \cyc(\ord_{ d_i}(t), \mathcal T_{\gcd_t(\nu_i)})\right)\oplus \mathcal G,$$
where $\mathcal G = \left\{ \begin{tabular}{ll}
$\{ \frac{q(q+1)}{2}\cdot \mathcal T_{\gcd_t(\nu_1)}+\frac{q(q-1)}{2}\cdot \mathcal T_{\gcd_t(\nu_2)}+\frac{q^2-1}{p-1}\cdot \mathcal T_{(p)}\}$ & $\textrm{if }p\mid t$; \\
$\frac{(q^2-1)}{\ord_{p}(t)}\times \cyc(\ord_{p}(t)) \oplus \{ \frac{q(q+1)}{2}\cdot \mathcal T_{\gcd_t(\nu_1)}+\frac{q(q-1)}{2}\cdot \mathcal T_{\gcd_t(\nu_2)} \}$ & $\textrm{if }p \nmid t$.\end{tabular}    \right.$


\end{corollary}

We provide a numerical example, showing the applicability of the previous corollary.

\begin{example}    
Consider the map $\varphi_2:\PGL(2, 5)\to \PGL(2, 5), g\mapsto  g^2$. From Corollary~\ref{cor:pgl}, we obtain that
$G(\varphi_2/\PGL(2, 5))$ is isomorphic to 

$10\times \cyc(2, \mathcal T_{(2)})\oplus 6\times \cyc(4) \oplus \{ 15\cdot \mathcal T_{(2, 2)}+10\cdot \mathcal T_{(2)}\}$; see Figure~\ref{fig:PGL}.


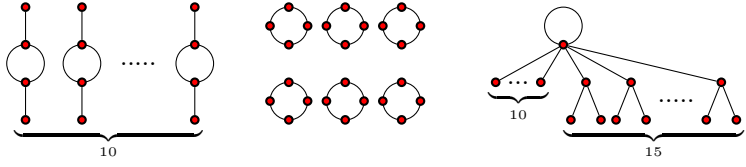
\begin{figure}
\begin{center}
\begin{tikzpicture}
  [scale=1, place/.style={circle,draw=black,thick,fill=red,
                 inner sep=0pt,minimum size=1mm}]
\draw (0,0) circle (2.5mm);
  \node (1) at (0, 0.25) [place] {};
  \node (2) at (0, -0.25) [place] {};
  \node (11) at (0, 0.75) [place] {};
  \node (21) at (0, -0.75) [place] {};
\foreach \x/\y in { 2/21, 1/11}
	\draw (\x) to (\y); 

\begin{scope}[xshift=0.75cm, yshift=0cm]
\draw (0,0) circle (2.5mm);
  \node (1) at (0, 0.25) [place] {};
  \node (2) at (0, -0.25) [place] {};
  \node (11) at (0, 0.75) [place] {};
  \node (21) at (0, -0.75) [place] {};
\foreach \x/\y in { 2/21, 1/11}
	\draw (\x) to (\y); 
\end{scope}

\node (ptos1) at (1.5,0) {.....};

\begin{scope}[xshift=2.25cm, yshift=0cm]
\draw (0,0) circle (2.5mm);
  \node (1) at (0, 0.25) [place] {};
  \node (2) at (0, -0.25) [place] {};
  \node (11) at (0, 0.75) [place] {};
  \node (21) at (0, -0.75) [place] {};
\foreach \x/\y in { 2/21, 1/11}
	\draw (\x) to (\y); 
\end{scope}

\node (underbrace1) at (1.1,-1.05) {\scriptsize $\underbrace{\hspace{2.5cm}}_{10}$};

\begin{scope}[xshift=3.5cm, yshift=0.5cm]
\draw (0,0) circle (2.5mm);
  \node (1) at (0, 0.25) [place] {};
  \node (2) at (0, -0.25) [place] {};
  \node (3) at (0.25, 0) [place] {};
  \node (4) at (-0.25, 0) [place] {}; 
\draw (0.75,0) circle (2.5mm);
  \node (1) at (0.75, 0.25) [place] {};
  \node (2) at (0.75, -0.25) [place] {};
  \node (3) at (1, 0) [place] {};
  \node (4) at (0.5, 0) [place] {}; 
\draw (1.5,0) circle (2.5mm);
  \node (1) at (1.5, 0.25) [place] {};
  \node (2) at (1.5, -0.25) [place] {};
  \node (3) at (1.75, 0) [place] {};
  \node (4) at (1.25, 0) [place] {}; 
\end{scope}

\begin{scope}[xshift=3.5cm, yshift=-0.5cm]
\draw (0,0) circle (2.5mm);
  \node (1) at (0, 0.25) [place] {};
  \node (2) at (0, -0.25) [place] {};
  \node (3) at (0.25, 0) [place] {};
  \node (4) at (-0.25, 0) [place] {}; 
\draw (0.75,0) circle (2.5mm);
  \node (1) at (0.75, 0.25) [place] {};
  \node (2) at (0.75, -0.25) [place] {};
  \node (3) at (1, 0) [place] {};
  \node (4) at (0.5, 0) [place] {}; 
\draw (1.5,0) circle (2.5mm);
  \node (1) at (1.5, 0.25) [place] {};
  \node (2) at (1.5, -0.25) [place] {};
  \node (3) at (1.75, 0) [place] {};
  \node (4) at (1.25, 0) [place] {}; 
\end{scope}

\begin{scope}[xshift=7.15cm, yshift=0.25cm]
\draw (0,0.25) circle (2.5mm);
  \node (1) at (0,0) [place] {};
  \node (2) at (0.3,-0.5) [place] {};
  \node (3) at (-0.3,-0.5) [place] {};
  \node (21) at (0.1,-1) [place] {};
  \node (22) at (0.5,-1) [place] {};
  \node (4y5) at (-0.6,-0.5) {...};
  \node (underbrace1) at (-0.6, -0.8) {\scriptsize $\underbrace{\hspace{0.8cm}}_{10}$};
  \node (6) at (-0.9,-0.5) [place] {};  
  \node (7) at  (0.9,-0.5) [place] {};
  \node (71) at (0.7,-1) [place] {};
  \node (72) at (1.1,-1) [place] {};
  \node (ptos) at (1.5,-0.75) {.....};
  \node (8) at  (2.1,-0.5) [place] {};
  \node (81) at (1.9,-1) [place] {};
  \node (82) at (2.3,-1) [place] {};
  \node (underbrace2) at (1.2, -1.3) {\scriptsize $\underbrace{\hspace{2.4cm}}_{15}$};
\foreach \x/\y in {1/2, 1/3, 1/6, 1/7, 1/8, 2/21, 2/22, 7/71, 7/72, 8/81, 8/82}
	\draw (\x) to (\y);
\end{scope}

\end{tikzpicture}
\end{center}
\caption{The graph  $10\times \operatorname{Cyc}(2, \mathcal T_{(2)})\oplus 6\times \operatorname{Cyc}(4) \oplus \{ 15\cdot \mathcal T_{(2, 2)}+10\cdot \mathcal T_{(2)}\}$.}
\label{fig:PGL}
\end{figure}


\end{example}

From Remark \ref{RemarkNumberOfTrees} a functional graph $\mathcal{G}(\varphi_t/PGL(2,q))$ has at most four non isomorphic trees. In the example above we have three non isomorphic trees but there are examples with four (for example the functional graph of $\varphi_2$ over $\PGL(2,11)$).

\mqureshi{
\section{Closing remarks}\label{Section:ClosingRemarks}

In this paper we describe the functional graph $\mathcal{G}(\varphi_t/G)$ when $G$ is an abelian group or a flower group. In both cases the cyclic part is easier to describe than the non cyclic part (i.e. the structure of the tree attached to periodic points). In contrast with the abelian case where all the trees attached to periodic points are isomorphic, for flower groups we proved that several non-isomorphic classes of trees can appear (this number depends on the cardinality of the petals and $t$). However for the families of groups considered in Section \ref{Section:applications} the number of non isomorphic trees in the functional graph $\mathcal{G}(\varphi_t/G)$ is at most four. We raise the following questions: 
\begin{itemize}
\item Is this number unbounded for general groups?
\item Is this number unbounded if we restrict to flower groups?
\item Determine necessary and sufficient conditions for the sequence $(c_0; c_1,\ldots,c_k)$ being the type of some flower group. 
\end{itemize}

These question are stated in increasing order of difficult. Relating to the above questions it is natural to consider the number $\tau(n)$, the maximum number of non isomorphic trees that can appear in a graph $\mathcal{G}(\varphi_t/G)$ for some group $G$ of order $n$ (restricted or not to flower groups) and some positive integer $t$. It could be interesting to determine the asymptotic behavior of the sequence $\tau(n)$.}

\section*{Acknowledgments}
Part of this paper were developed during a pleasant stay by the second author at Universidad de la Rep\'ublica, supported by PEDECIBA. \mqureshi{The authors thank S\'avio Ribas to pointing out a mistake in Lemma \ref{lem:semi} in a previous version of this paper}.


\begin{thebibliography}{99}

\bibitem{ahmad} U.~Ahmad.
\newblock{The power digraphs associated with generalized dihedral groups}.
\newblock{\em Discrete Math. Algorithms Appl.} 7(4): 1550057 (2015).

\bibitem{a2} U.~Ahmad and M. Moeen. 
\newblock {The digraphs arising by the power maps of generalized
Quaternion groups}. 
\newblock {\em J. Algebra Appl.} 16(9): 1750179 (2017).


\bibitem{BBS} L.~Blum, M.~Blum and M.~Shub.
\newblock {A simple unpredictable pseudo-random number generator}.
\newblock {\em SIAM J. Comput.} 15(2): 364--383 (1986). 





\bibitem{deng} G. Deng and J. Zhao.
\newblock{Digraph from power mapping on noncommutative groups}.
\newblock {\em J. Algebra Appl.} 19(5): 2050084 (2020).

\mqureshi{
\bibitem{deKlerketal} de Klerk, B-E., and Johan H. Meyer. 
\newblock{Functional graphs of abelian group endomorphisms}. 
\newblock{\em Discrete Mathematics} 345.2: 112691 (2022).
}

\bibitem{Gassert14}
T.~A.~Gassert. 
\newblock Chebyshev action on finite fields. 
\newblock {\em Discr. Math.} 315: 83--94 (2014). 




\bibitem{cent} S.M.~Jafarian Amiri and H. Rostami.
\newblock Finite groups whose all proper centralizers are cyclic 
\newblock{\em Bull. Iranian Math. Soc.} 43(3): 755-762 (2017).




\bibitem{larsen}M. Larson. 
\newblock{Power maps in finite groups}. 
\newblock{\em Integers} 19: \#A58 (2019).



\bibitem{MPQ19}
R.~Martins, D.~Panario and C.~Qureshi.
\newblock A Survey on Iterations of Mappings over Finite Fields.
\newblock In: {\em Combinatorics and finite fields: Difference sets, 
polynomials, pseudorandomness and applications.} Vol.~23 Radon Series on Computational and Applied Mathematics, De Gruyter, Berlin (2019).





\bibitem{MV88} G.L. Mullen and T.P. Vaughan. 
\newblock Cycles of linear permutations over a finite field. 
\newblock {\em Linear Algebra Appl.} 108: 63-82 (1988).






\bibitem{PR18} D.~Panario and L.~Reis.
\newblock The functional graph of linear maps over finite fields and applications.
\newblock{\em Des. Codes Cryptogr.} 87(2), 437--453 (2019).


\bibitem{PMMY01}
A.~Peinado, F.~Montoya, J.~Munoz and A.~J.~Yuste.
\newblock Maximal periods of $x^2+c$ in $\F_q$. 
\newblock {\em In: International Symposium on Applied Algebra, Algebraic Algorithms and Error-Correcting Codes, Springer} pp.~219--228 (2001).

\bibitem{P75} J.~M.~Pollard.
\newblock {A Monte Carlo method for factorization.}
\newblock {\em BIT} 15(3): 331--334 (1975).


\bibitem{QP15}
C. Qureshi and D. Panario. 
\newblock R\'edei actions on finite fields and multiplication map in cyclic groups. 
\newblock {\em SIAM J. on Discr. Math.} 29: 1486--1503 (2015).

\bibitem{QP18}
C.~Qureshi and D.~Panario.
\newblock The graph structure of Chebyshev polynomials over finite fields and applications.
\newblock {\em Des. Codes Cryptogr.} 87(2), 393--416 (2019).

\bibitem{QPM17} 
C. Qureshi, D. Panario and R. Martins. 
\newblock Cycle structure of iterating R\'edei functions. 
\newblock {\em Adv. Math. Comm.} 11(2): 397--407 (2017).

\bibitem{QR19} C. Qureshi and L. Reis.
\newblock Dynamics of the $a$-map over residually finite Dedekind domains.
\newblock{\em J. Num. Theory} 204: 134--154 (2019).

\bibitem{R19}
L.~Reis.
\newblock Moebius-like maps on irreducible polynomials and rational transformations
\newblock{\em J. Pure Appl. Algebra} 224(1): 169--180 (2020).

\bibitem{Rogers96}
T.~Rogers. 
\newblock The graph of the square mapping on the prime fields. 
\newblock {\em Discr. Math.} 144: 317--324 (1996).

\bibitem{T05} R.~A.~H.~Toledo, 
\newblock  Linear Finite Dynamical Systems. 
\newblock {\em Commun. Algebra.}  33 (9): 2977-2989 (2005).

%
%
\bibitem{Ugolini18}
S.~Ugolini. 
\newblock Functional graphs of rational maps induced by endomorphisms of ordinary elliptic curves over finite fields.
\newblock {\em Periodica Math. Hungarica} 77.2: 237--260 (2018).


\bibitem{VS04}
T. Vasiga and J. Shallit.
\newblock On the iteration of certain quadratic maps over $GF(p)$. 
\newblock {\em Discr. Math.} 277: 219--240 (2004).
%

\bibitem{WZ98} 
M. Wiener and R. Zuccherato. 
\newblock Faster attacks on elliptic curve cryptosystems. 
\newblock {\em International workshop on selected areas in cryptography. Springer, Berlin, Heidelberg} 190--200 (1998).















%
%

%











\end{thebibliography}
\end{document}